\documentclass[10pt]{article}
\usepackage[english]{babel}
\usepackage[paper=letterpaper,top=2.5cm,bottom=2.5cm,inner=2.5cm,outer=2.5cm]{geometry}
\usepackage{setspace}

\date{}

\renewcommand{\baselinestretch}{1.48}

\newcommand{\bpi}{ {\mbox{\boldmath $\pi$}} }

\usepackage{amsmath, amsthm, amssymb}

\usepackage{widetext}

\usepackage{algorithm}

\usepackage{algorithmic}

\usepackage{bm}
\usepackage{url}
\usepackage{comment}

\usepackage[dvips]{graphicx}

\newsavebox{\ggbox}

\newcommand{\displaymathbox}[1]%
{\savebox{\ggbox}{$\displaystyle #1$}\begin{displaymath}\fbox{\usebox{\ggbox}}%
\end{displaymath}}

\newcommand{\eat}[1]{}

\newcommand{\equationbox}[2]%
{\savebox{\ggbox}{$\displaystyle #1$}\begin{equation}\fbox{\usebox{\ggbox}}%
\label{#2}\end{equation}}

\newcommand{\myapp}[1] {Appendix~\ref{#1}}

\newcommand {\beq} {\begin{equation}}
\newcommand {\eeq} {\end{equation}}
\newcommand {\bear} {\begin{eqnarray}}
\newcommand {\eear} {\end{eqnarray}}
\newcommand {\bears} {\begin{eqnarray*}}
\newcommand {\eears} {\end{eqnarray*}}

\newtheorem{assumption}{Assumption}[section]
\newtheorem{definition}{Definition}[section]

\newtheorem{proposition}{Proposition}[section]

\newtheorem{remark}{Remark}[section]
\newtheorem{corollary}{Corollary}[section]

\def\square{\vrule height6pt width7pt depth1pt}

\def\R{\mathop{\rm I\kern -0.20em R}\nolimits}

\def\b1{{\bf 1}}

\def\be{{\bf e}}

\def\cost{{G}}

\def\nn#1{\mathop{\left |\kern -0.10em \left | #1
        \right |\kern -0.10em \right |_\mu}\nolimits}

\def\utility{{U}}

\begin{document}

\title{Forever Young: Aging Control For Smartphones in Hybrid Networks}

\author{
{Eitan Altman}\\
{INRIA \& University of Avignon, France}
\and
{Rachid El-Azouzi} \\
{University of Avignon, France}
\and
 {Daniel S. Menasch\'{e} } \\
{UMass, USA; UFRJ, Brazil} 
%
\and
{Yuedong Xu}\\
{INRIA \& University of Avignon, France} 
} 

\maketitle

\begin{abstract}
The demand for Internet services that require frequent  updates through small messages, such as  microblogging,  has tremendously grown in the past few years. Although the use of such  applications by domestic users is usually free, their access from mobile devices   is subject to fees and consumes energy from limited batteries.  If a user activates his mobile device and  is in range of a service provider,  a content update is received at the expense of    monetary and energy costs. Thus, users face a tradeoff between such costs and their messages aging. The goal of this paper is to show how to cope with such a tradeoff, by devising \emph{aging control policies}.  An aging control policy consists of deciding, based on the current utility of the  last message received,   whether to activate the  mobile device, and if so, which technology to use (WiFi or 3G). We present a model that yields the optimal aging control policy.   Our model is based on a Markov Decision Process in which states correspond to message ages.  Using our model, we show  the existence of an optimal strategy in the class of  threshold strategies, wherein users activate their mobile devices if the age of their messages surpasses a given threshold and remain inactive otherwise.   We then consider strategic content providers (publishers) that offer \emph{bonus packages} to users, so as  to incent them to download updates  of  advertisement campaigns.  We provide simple algorithms  for publishers to determine optimal bonus levels, leveraging the fact that users adopt their   optimal aging control strategies.  The accuracy of our model is validated against traces from the UMass DieselNet bus network.  
   \end{abstract}

\section{Introduction}

The demand for Internet services that require frequent  updates through small messages has tremendously grown in the past few years.  While the popularity of traditional applications of that kind, such as weather forecasts, traffic reports  and  news, is unlike to decline,  novel applications, such as Twitter~\cite{twitter},  have arisen.       Twitter, alone, recorded a 1,500\% increase in the number of registered users since 2006, and currently counts with more than 100 million users worldwide.    For a second example, mobile messaging with Exchange ActiveSync~\cite{exchange} allows smartphone users to receive timely updates when new data items arrive in their mailboxes.   Today, Exchange ActiveSync is supported by more than 300 million mobile devices worldwide.  Henceforth, for the sake of concreteness we focus on microblogging applications, such as Twitter or similar news feeds.  

Users of microblogging applications join interest groups and aim at receiving small messages from editors.  As messages age, they get outdated and their utilities decrease. As a consequence, users must control when to receive updates.  A user willing to receive an update activates his mobile device, which then broadcasts periodic beacons to inform demands to  service providers.

Although the use of microblogging  applications by domestic users is usually free, their access from mobile devices  consumes energy from limited batteries  and is subject to fees.   We consider users that can access the Internet either through a WiFi or 3G network.  The 3G network provides broader coverage, but its usage requires a subscription to a cell phone data plan.   More generally, our results apply to any scenario in which users can access the network through multiple interfaces with different costs and ubiquitousness~\cite{wifler, banerjee}.  

 Let the \emph{age} of a message held by a user be the duration of the interval of time since   the message was downloaded by such user.  
 If a user activates his mobile device and  is in the range of a service provider (WiFi access point or 3G antenna),   an update is received and the age of the message held by the user is reset to one,  at the expense of  the previously mentioned monetary and energy costs.  
  Thus, users face a tradeoff between energy and monetary costs and their messages aging.  To cope with such a tradeoff, users decide, based on the age of the stored message,  whether to activate the  mobile device, and if so, which technology to use (WiFi or 3G).   
We refer to a  policy which determines  activation decisions  as a function of message ages as an \emph{aging control policy}.  
\emph{The first goal of this paper is to devise efficient aging control policies.}

Strategic content providers can incent users to download updates of advertisement campaigns or unpopular content by offering \emph{bonus packages}.  The goal of the bonus package, translated in  terms of our aging control problem, consists of minimizing the average age of content held by users, subject to a budget on the number of messages transmitted  per time slot, as dictated by the service provider capacity.    Although nowadays  bonus packages are set exclusively by  service providers~\cite{sfr},  we envision that  in the future content providers will  reach agreements with service providers.   Through such agreements, content providers, also  known as publishers,  will play an important  role in the settlement of  bonus packages~\cite{misra}.    \emph{The second goal of this paper is to solve the  publishers' bonus selection problem, when  users adopt an aging control policy.  }  

 We pose the two following questions,
\begin{enumerate}
\item what is the users optimal aging control policy?
\item leveraging the users optimal aging control policy, what is the publishers optimal bonus  strategy?
\end{enumerate}

We propose a  model that allows us to answer the  questions above.  Our model  accounts for energy costs,  prices and the utility of messages as a function of their age.  Using our model, we show that users can maximize their utilities by adopting a simple threshold policy.  The policy consists of activating the mobile device if the content age surpasses a given threshold and remaining  inactive otherwise.     We  derive properties of the optimal threshold, and a closed-form expression for the average reward obtained by users as a function of the selected strategies.   We then show the accuracy  of our approach using traces collected from the UMass DieselNet bus network.   Using traces, we also study  location-aware policies, according to which users can activate their devices based on their position on campus, and compare them against location-oblivious policies.

For the strategic publishers, we present two simple algorithms to solve the bonus determination problem posed above.  The first algorithm presumes complete information while the second consists of a learning algorithm for  publishers that have imperfect information about the system parameters, and is validated using trace-driven simulations.    Finally, we  show the convergence of the proposed learning algorithm, making use of  results on differential inclusions and stochastic approximations.

  In summary, we make the following contributions.

\textbf{Model formulation: } We introduce the aging control problem, and propose a model to solve it.  Using the model, we  derive  properties about the optimal aging control policy and closed-form expressions for the expected average reward. 

\textbf{DieselNet trace analysis: } We quantify how aging control policies impact users of the DieselNet bus network.  Using traces collected from DieselNet, we show the accuracy of our model estimates and analyze policies that are out of the scope of our model.

\textbf{Mechanism design: }  We provide two simple algorithms for publishers to incent users to download advertisement updates,  leveraging the fact that users adopt their   optimal aging control strategies.   We formally show the convergence of the  proposed learning algorithm, and numerically investigate its accuracy and convergence rate using trace-driven simulations. 

The remainder of this paper is organized as follows.   After further discussing  the need for aging control, in~\textsection\ref{model} we present  our model, and in~\textsection\ref{eval} we report  results obtained with traces from the UMass DieselNet bus network.  Our model analysis is shown in~\textsection\ref{optimal}, followed by the  solution of the publishers problem in~\textsection\ref{maker}.
We present  a discussion of our modeling assumptions in~\textsection\ref{limitations}, related work in ~\textsection\ref{related} and ~\textsection\ref{conclusion} concludes the paper.

\section{Why Aging Control?}
\label{sec:why}
The goal of an aging control policy is to provide high quality of service while   1) reducing energy consumption and 2) reducing 3G usage, by leveraging WiFi connectivity when available.     Whereas the reduction in energy consumption is of interest mainly to subscribers,  the reduction in 3G usage is of interest both to service providers and subscribers.  Next, we discuss a few issues related to the adoption of aging control from the service providers and subscribers standpoints.

\subsection{Service Provider Standpoint: Limited Spectrum}

The increasing demand for mobile Internet access is creating pressure on the service providers, whose limited spectrum might not be sufficient to cope with the demand~\cite{anger, wifler}.    To deal with such pressure, some wireless providers are offering incentives to subscribers to reduce their 3G usage by switching to WiFi~\cite{tmobile}.  Therefore, it is to the best interest of the service providers to devise efficient aging control policies for their users.  Aging control policies can not only reduce the pressure on the 3G spectrum, but also reduce the costs to service providers that support hybrid 3G/WiFi networks~(with savings of up to 60\% against providers that offer only 3G~\cite{wifler}).  

\subsection{Users Standpoint: Energy Consumption, Monetary Costs and Coverage}

How do the coverage, energy consumption and monetary costs vary between 3G and WiFi?  Table~\ref{tab:comp}  illustrates the answer with data obtained from~\cite{wifler, niranjan, sms, link} and is further discussed next.  

\begin{table}
\center
\begin{tabular}{l|l|l|l}
\hline
& Coverage & Energy Consumption (Scanning and Handshaking) & Monetary Costs \\
\hline
\hline
WiFi & 10\%-20\%   & 0.63J (Android G1) 1.18J (Nokia N95) & free \\
\hline
3G   & 80\%-90\% &  $\approx 0$   &  0.20 USD (SMS in USA) \\
\hline
\end{tabular}
\caption{Coverage, energy consumption and monetary costs}
\label{tab:comp}
\end{table}

\paragraph{Energy consumption and AP scanning}  The energy efficiency of WiFi and 3G radios differs significantly.   WiFi users actively scan for APs by broadcasting probes and waiting for beacons sent by the APs.  If multiple beacons are received, users connect to the AP with highest signal strength.  The association overhead, incurred when searching and connecting to the APs, yields substantial energy costs.    Niranjan \emph{et al.}~\cite{niranjan} report that the energy consumption of a WiFi AP scan is 0.63J in Android G1 and 1.18J in Nokia N95,  taking roughly 1 and 2 seconds, respectively (note that in one time slot users might perform multiple AP scans).    In the remainder of this work, to simplify presentation, we consider only the energy costs related to association overhead.  In many scenarios of practical interest WiFi and 3G radios  incur comparable energy costs after the handshaking phase, and our model can be easily adapted to account for cases in which such energy costs differ.

\paragraph{Monetary costs} Whereas free WiFi hotspots are gaining popularity~\cite{paris, amherst}, the use of 3G is still associated to monetary costs.   For example, in the United States and in Australia, it typically costs 0.20 USD to send an SMS  through a 3G network~\cite{sms}.  As a second example, in Brazil subscribers of Claro~3G incur a cost of USD 0.10 per megabyte after exceeding their monthly quota~\cite{gizmodo}.   

\paragraph{Coverage} It is well known that the coverage of 3G is much broader than the coverage of WiFi.  This is because the 3G towers are placed by operators so as to achieve almost perfect coverage, while WiFi access points are distributed  and activated in an ad hoc fashion.    For instance, it has been reported in \eat{by Balasubramanian \emph{et al.}}~\cite{wifler} that 3G and WiFi are available roughly  12\% and 90\% of the time in the town of Amherst, Massachusetts. 

\paragraph{Approximations}  In light of the above observations, in the rest of this paper, except otherwise stated, we consider the following three approximations concerning energy consumption, monetary costs, and coverage.   We assume that 1) the energy consumption of scanning for WiFi access points dominates the energy costs of the 3G and WiFi radios, 2) WiFi access points are open and freely available, whereas the use of 3G incurs a monetary cost per message and 3) WiFi is intermittently available whereas 3G offers perfect coverage.  Note that the above three approximations are  made solely to simplify presentation:  our model has  flexibility to account for  different energy consumption and monetary costs of WiFi and 3G, and can be easily adapted to account for a 3G network that does not have perfect coverage.

\subsection{Adopting Aging Control}

Next, we discuss two key aspects pertaining the adoption of aging control: 1) the delay tolerance of the applications and 2) the availability of WiFi access points.  

\subsubsection{Delay Tolerant Applications: Age and Utility}

Aging control is useful for applications that can tolerate delays, such as news feeds and email.   For these applications, users might be willing to tolerate some delay if that translates into reduced energy consumption or monetary costs.  

\paragraph{Workload}  \label{sec:workload}    The  workload subject to aging control can be due  to one or multiple applications.  

\textbf{One application: }  Websites such as Yahoo!  provide news feeds on different topics,  ranging from  business and economics to entertainment.  Niranjan \emph{et al.}~\cite{niranjan} monitored feeds in ten categories, and reported that in three of these categories (business, top stories and opinion/editorial) at least one new feed is available every minute, with very high probability~\cite[Figure 13]{niranjan}.   \label{sec:heavyload}

\textbf{Mutiple applications: }
Aiming at energy savings, it is common practice for users of smartphones to synchronize the updates of multiple applications, using APIs such as OVI Notifications~\cite{ovi}. In this case, the larger the number of applications that require updates, the higher the chances of at least one update being available every time slot.  

In the remainder of this paper, we focus on the above mentioned workloads, for which updates are issued with high frequency.  The analysis of light and medium workloads is subject for future work (see \textsection\ref{limitations}).

\paragraph{User Interface}  

How does the utility of a content degrade with its age?  The answer to this question is user and application dependent.  It turns out that some applications, such as JuiceDefender~\cite{juicedefender}, already allow iPhone users to specify their delay tolerance per application so as to save battery.  For instance, users can  simply specify a delay threshold, after which they wish to have received an update.  In our model, this corresponds to an utility that has the shape of a step function (see \textsection\ref{sec:special}).

\subsubsection{WiFi Access Points}

Open WiFi access points are  widespread~\cite{paris, amherst}. The age control mechanisms described in this paper can rely on such access points, which users encounter in an ad hoc fashion.   Alternatively, Internet service providers can deploy their own  closed networks of WiFi access points, in order to alleviate the load of their 3G networks.   In France, enterprises such as SFR and Neuf are already adopting such strategy. The age control mechanisms presented in this paper can be applied in this setting as well.  In this case, users are required to perform a web-based login in order to access the WiFi spots, which can be done automatically using tools such as Freewifi~\cite{freewifi}.  Note that even if mobile users have perfect geographic information about the location of the access points, random factors like attenuation (fading) and user speed will determine the availability of the access points, which will  vary temporally and spatially.    

\paragraph{Markov Decision Processes}   Remaining inactive at a given instant of time is particularly useful if future transmission opportunities are available in  upcoming time slots.   However, as discussed above, WiFi access points are intermittently available.  Therefore, WiFi users experience a random delay between the instant  of time at which they activate  their devices and the instant at which they obtain content updates.    How to account for the unpredictability of transmission opportunities in upcoming slots?  To answer this question, in the next section we propose a model, based on a Markov Decision Process,  which  naturally accounts for the effect of the actions 
at a given time slot on the future states of the system, and allows us to derive the structure of the optimal aging control policy.

\section{Model}

\label{model}

We  consider  mobile users that  subscribe to receive content updates.  In the microblogging jargon, such users are said \emph{to follow} a content.  Content is transmitted from publishers to users, through messages sent by the service providers. The \emph{age} of the message held by a user   is defined as the length of time, measured in time slots, since the message was downloaded.   Note that we assume that updates are available at every time slot with high probability (see ~\textsection\ref{limitations}).

 Let $x_{t}$ be the age of the message held by a user  at time $t$.  The age of the message equals one when the user first receives it,  and increases by one every time slot, except when the user obtains an update, time at which the age is reset to one.  

A user can receive message updates when  his mobile device is in  range of a service provider and the contact between them lasts for a minimum amount  time which characterizes a \emph{useful contact opportunity}.    While the  3G technology is assumed to guarantee perfect coverage, WiFi users are subject to outages.  
 Let $e_{t}$ be an indicator random variable equal to 1 if there is a useful contact opportunity with a  WiFi provider at time $t$, and 0 otherwise.  We let $p=E[e_t]$ and assume $0<p<1$.  Next, we state our key modeling assumption.
  \begin{assumption} \label{assum:uniform}
 \textbf{Uniform and independent contact opportunity distribution: } The probability of a useful contact opportunity between a user and    WiFi providers  is  constant and independent across time slots,  and equals~\underline{$p$}.
 \end{assumption}
 
 Under Assumption~\ref{assum:uniform}, there are no correlations in time between contact opportunities experienced by a user; this is  a strong assumption, since such correlations are present  in any mobile network, as illustrated in~\textsection\ref{sec:measurecontact}.  However, as shown in  ~\textsection\ref{sec:evaluating}, there are scenarios of practical interest in which the uniformity and independence assumption does not compromise the accuracy of the results obtained using our model.  Therefore, we proceed our analysis under such an assumption, indicating its implications throughout the paper and studying scenarios that are out of the scope of our model using real-world traces.

\label{para:thrpol} 
     The goal of each user is to minimize the expected age of the  content he follows, accounting for  energy and monetary costs.   In order to achieve such goal, users must choose at each time slot $t$ their actions, $a_{t}$.  The available actions are 
\begin{equation}
a_{t}=
\left\{ \begin{array}{ll} 0, & \textrm{inactive}    \\
1, & \textrm{WiFi }    \\
2,  & \textrm{WiFi if $\exists$ useful contact opportunity, else 3G } \nonumber \end{array} \right.
\end{equation}

\paragraph{State Dynamics}
The age of the content followed  by a user increases by one if his device is inactive or if it is not in  range of  a service provider, and is reset to one otherwise.  Let ${M}$ be the maximum age of a content.  Then, 
\begin{equation}
x_{t+1} {=}
\left\{ \begin{array}{ll} \min(x_{t}+1,{M}),  & \textrm{ if }  ( a_{t} {=} 0 ) \textrm{ or  } (a_{t} {=} 1 \textrm{ and  } e_{t} {=}0) \\
1,  & \textrm{ if }   (a_{t} {=} 2 ) \textrm{ or } (a_{t} {=} 1 \textrm{ and } e_{t} {=}1) \end{array} \right.
\end{equation}

\paragraph{Utility}

Let $U(x_{t})$  the utility of the  followed content at time $t$. We assume that $U(x)$ is a non-increasing function of $x$, which corresponds to messages that become obsolete with time, and that $U(x)=Z$ if $x \ge M$.

\paragraph{Costs}

Let $G$ be the cost incurred to maintain the mobile device active, measured in monetary units.   Then, the energy cost $c_t$ is given as a function of $a_{t}$ as
\begin{equation} \label{energycost}
c_{t}(a_{t})=
\left\{ \begin{array}{ll}  \cost, & \textrm{ if }   a_{t} \ge 1\\
0,  & \textrm{ if } a_{t}=0 \end{array} \right.
\end{equation}

Service providers charge a price for each message  transmitted.  
The prices charged by WiFi and 3G providers are $P$ and $P_{3G}$, respectively.  
When a user receives an update, he is subject to a monetary cost of $m_t$,
\begin{equation} \label{energycost}
m_{t}(a_{t},e_t){=}
\left\{ \begin{array}{ll}  P_{3G}, & \textrm{ if }   a_{t} = 2 \textrm{ and } e_t = 0\\
P,  & \textrm{ if } a_{t} \ge 1 \textrm{ and } e_{t} =1 \end{array} \right.
\end{equation}

Content providers, also referred to as publishers, can offer bonuses to users that follow their contents.  Such bonuses are set in agreement with service providers, and are transfered to users as credits.  Let $B_t$ be the bonus level set by the content provider, $B_t \leq \min(P_{3G}, P)$.

   The instantaneous user reward at time $t$, $r_{t}(x_{t},a_{t})$, is
\begin{equation} \label{rtiddef} 
r_{t}(x_{t},a_{t}){=} 
 \utility(x_{t}) {-} c_t(a_t)  {-}  \max(m_t{-}B_t,0)
\end{equation}

\paragraph{Users Strategies}   The \emph{strategy} of a user is given by the probability of choosing a given action $a_t$ at each of the possible states.   Without loss of generality, in this paper we restrict to Markovian stationary policies~\cite[Chapter 8]{puterman}. The probability of choosing action $a_t$ at state $x_t$ is denoted by $u(a_t|x_t)$.  

\paragraph{Problem Definition} The problem faced by each user consists of finding the strategy $u$ that  maximizes his  expected average reward.

\medskip
\boxed{
\begin{minipage}{6.3in}
{\ensuremath{\mbox{\sc  User Problem:}}} 
Obtain strategy $u$ so as to maximize $E[r;u]$, where  \vspace{-0.1in}
\begin{equation}   \label{def:eri}
E[r;u] = \lim_{\ell \rightarrow \infty} \frac{1}{\ell} \sum_{t=0}^{\ell} E[r_{t}(x_{t},a_{t});u]  
\end{equation}
\end{minipage}
}
\medskip

In what follows, we drop the subscript $t$ from variables when analyzing the system in steady state.

\begin{figure}
\center 
\begin{tabular}{cc}
\includegraphics[scale=0.40]{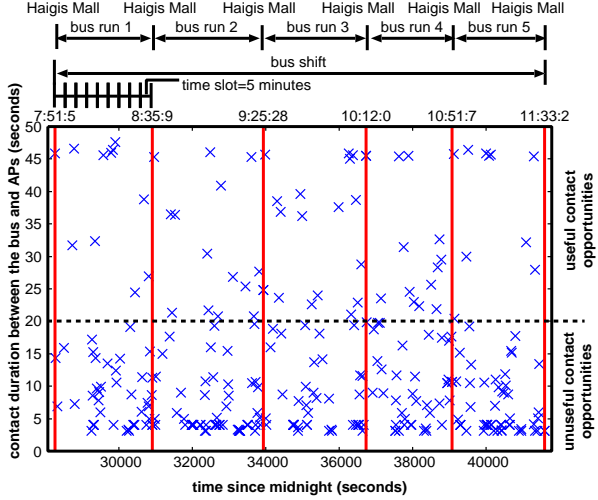} & 
\includegraphics[scale=0.60]{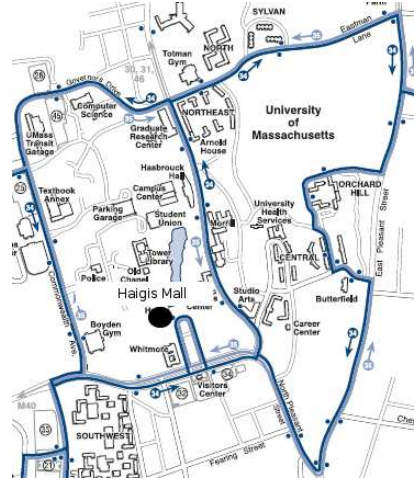}\\
(a) Bus-AP contacts during a typical bus shift &
(b) Route of a bus that passes through Haigis Mall
\end{tabular}
\caption{A typical bus shift.}
\label{history}
\vspace{-0.18in}
\end{figure}

\paragraph{Optimal Threshold Policy}  Next, we introduce two special classes of policies, the \emph{two-threshold policies} and the \emph{threshold policies}. Let $s$ and $s_{3G}$ be the WiFi  and the 3G thresholds, $1 \leq s \leq  s_{3G} \leq M+1$.  A policy which consists of setting $a=0$ when $x < s$, $a=1$ when $s \leq  x < s_{3G}$, and $a=2$ if $x \ge s_{3G}$, is referred to as a \emph{two-threshold policy}.  If users have no access to 3G $(P_{3G} \rightarrow \infty)$,  a 
policy which consists of backing off if $x < s$ and being active if $x \ge s$ is referred to as a \emph{threshold policy}.     

Note that while $x \in  [1, M]$, the thresholds assume values in the range $[1, M+1]$.  When $P_{3G} \rightarrow \infty$,  $s=M+1$ means that the user should remain always inactive  (refer to Table~\ref{notation} for notation). 

The following proposition reduces the problem of finding the optimal policy for the {\ensuremath{\mbox{\sc  User Problem}}} to the one of finding the two thresholds $s$ and $s_{3G}$. 
\begin{proposition} \label{temp1}
The {\ensuremath{\mbox{\sc  User Problem}}} admits an optimal policy  of two-threshold type such that
\\
\\
\begin{tabular}{p{4.6cm}|p{3.4cm}}
\textrm{if } $P_{3G} > G/p + P$,
\begin{equation*}
{a(x)} {=} \left\{ \begin{array}{ll}  0, & \textrm{ if }   x < s \\
1,  & \textrm{ if } s \leq {x} < s_{3G} \\
2, & \textrm{otherwise} 
\end{array} \right.
\end{equation*}
&
\textrm{if } $P_{3G} \leq G/p + P$,
\begin{equation*}
{a(x)} {=}
\left\{ \begin{array}{ll}  0, & \textrm{ if }   x {<} s_{3G} \\
2, & \textrm{otherwise} 
\end{array} \right.
\end{equation*}
\end{tabular}

If $P_{3G} \rightarrow \infty$, $s_{3G}=M+1$ and the {\ensuremath{\mbox{\sc  User Problem}}}  admits an optimal policy of threshold~type. 
\end{proposition}
When the price charged by the 3G providers dominates the costs, the optimal policy consists of activating the WiFi radio if $s \leq {x} < s_{3G}$ and using the 3G only if  $x \ge s_{3G}$.   If the price charge by the 3G providers is dominated by $G/p + P$, though, users are better off relying on the 3G technology, which offers  perfect coverage.  

To simplify  presentation, in the upcoming sections we assume that 1) users have access only to WiFi APs ($P_{3G} \rightarrow \infty$) and 2)  the optimal threshold policy is unique.  In~\myapp{sec:w3g} we 1) show how  our results extend to the scenario in which users can choose between WiFi and~3G and 2) characterize the scenarios under which the optimal threshold is not unique. \eat{Given that the optimal policy $u$ is fully characterized by the threshold $s$, we let  $E[r;s]=E[r;u]$.   Next, we state the definition of the optimal~threshold. }
\begin{definition} \label{defopt}
The threshold $s^{\star}$ of the optimal policy is 
\begin{equation} \label{optdef}
s^{\star} = \textrm{argmax}_s \{ E[r;s] \}
\end{equation}   
\end{definition}
The properties of the optimal threshold are discussed in the following section in light of traces from DieselNet, and are formally stated in~\textsection\ref{optimal}.

\section{Evaluating Aging Control Policies in DieselNet}

\label{eval}

In this section we use traces collected from the UMass Amherst DieselNet~\cite{Thedu-mobicom08}  to evaluate how aging control policies perform in practice.   Our goals are 1) to show the accuracy  of our model predictions, 2) to assess the optimality of the model predictions in the class of threshold policies and  \eat{in the class of threshold policies } and 3)  to compare the optimal policy  obtained  with our model against policies  that are out of the scope of our model, such as  location-aware policies.  

The users of the microblogging application considered in this section are  passengers and drivers of buses.  We assume that  users are  interested in following one of the heavy load news feeds described in \textsection\ref{sec:heavyload}, for which updates are issued every minute with high probability. \eat{street traffic information.  }  Time is divided into slots of 5 minutes. \eat{, during which a traffic update is made available by the transit authority.   }

We begin with  an overview of statistics collected from the traces and their implications on aging control.

\subsection{Measuring Contact Statistics}

\label{sec:measurecontact}

To characterize the update opportunities experienced by the users, we analyze contacts between buses and access points (APs) at the UMass campus. The traces were compiled during Fall 2007 from buses running routes serviced  by UMass Transit.
 Each bus scans for connection with  APs on the road, and when found, connects to the AP and records the duration of the connection~\cite{Thedu-mobicom08}.

 \begin{figure*}
\hspace{-0.1in}
\includegraphics[scale=0.31]{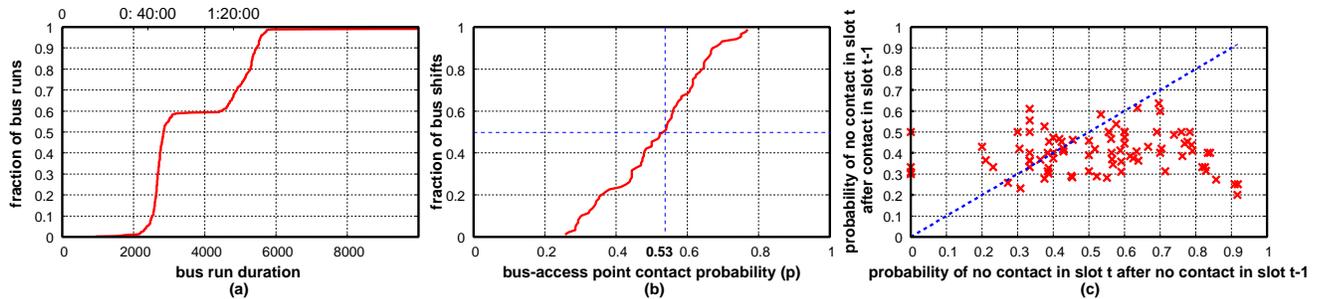}
\caption{Useful contact opportunity statistics (a) time  between contacts to AP at Haigis Mall; (b) CDF of p; (c) scatter plot of opportunities in consecutive slots.}  \label{intercontact2}  \label{intercontactb}   \label{intercontacta}

\end{figure*}

 Figure~\ref{history}(a)  illustrates a sample of the trace data.  In Figure~\ref{history}(a), each cross corresponds to a contact between the bus and an AP. The  $x$ and $y$ coordinates of each cross correspond to the time of the day at which the contact occurred and the duration of the contact, respectively.

  In the rest of this paper, we assume that when the mobile devices are active,  scans for  access points occur every 20 seconds, and last up to the moment at which an access point is found. Ra \emph{et al.}~\cite[\textsection 5]{link} empirically determined that a scanning frequency of 1/20 seconds yields a good balance between efficiency and low energy expenditure. An access point is considered useful once it is scanned in two consecutive intervals of 20 seconds.  \eat{ i.e., if the contact duration between the user and the access point lasts for at least 20 seconds. }
  \eat{a useful transmission between the bus and the AP requires 20 seconds.  Twenty seconds roughly correspond to the time required to  scan for access points, connect to  one of them and download a message of up to 500KBytes at a rate of 200Kbps (a standard  802.11g AP supports 6Mbps, which we assume  shared among 30 users).       }
  Henceforth, we refer to contact opportunities that last at least 20 seconds as \emph{useful contact opportunities}, or  \emph{contacts} for short, when the qualification is clear from the context.  Time slots in which at least one useful contact opportunity begins are referred to as \emph{useful~slots}.  

 Figure~\ref{history}(b) shows the map of one of the bus routes considered in this paper, with the Haigis Mall in evidence.   The Haigis Mall is a central location for the transit of the buses at UMass,  and in this work we restrict to routes that pass through it.   In Figure~\ref{history}(a), vertical lines correspond to instants at which the bus passed through the Haigis Mall.   A \emph{bus run} corresponds to the interval between two arrivals at Haigis Mall.  A \emph{bus shift} corresponds to a sequence of consecutive and uninterrupted bus runs, in the same day.   Note that a bus run in Figure~\ref{history}(a) takes around 40 minutes.    A typical bus run varies between 40 minutes and 1 hour and 20 minutes (see  Figure~\ref{intercontact2}(a)).

 Figure~\ref{history}(a) shows that roughly every time the bus passes through Haigis Mall there is a useful contact opportunity.  This observation has important implications on the aging control policy. In particular, it indicates that users that are location-aware can take advantage of such information in order to devise their activation strategies.  We will evaluate the performance of the policy which consists of activating the mobile device only at Haigis Mall in~\textsection\ref{locationaware}. 

Next, our goal is to study the distribution of contact opportunities between buses and APs. In particular, we wish to identify the extent at  which the uniformity and independence assumption (Assumption~\ref{assum:uniform}) holds in practice, specially when buses are far away from Haigis Mall (see Figure ~\ref{history}).  To this aim, for each day and for each bus shift, we generate out of our traces a \emph{string of ones and  zeros, corresponding to  useful and unuseful slots}, respectively.     Such strings are  used to plot  Figures~\ref{intercontactb}(b) and ~\ref{intercontactb}(c), as well as an input to our trace-driven simulator, described in the next section.

Figure~\ref{intercontactb}(b)  shows the CDF of the bus-AP contact probability ($p$) across bus shifts.  The median of $p$ is 0.53.  The contact probability is above 0.3 for up to 85\% of the bus shifts and below 0.7 for more than 90\% of the bus shifts.

Figure~\ref{intercontacta}(c)  shows a scatter plot of the contact probabilities in two consecutive slots. Each cross corresponds to a bus shift.   A cross with coordinates $x$ and $y$ corresponds to a bus shift in which the probability of no contact in slot $t$ after no contact in slot $t-1$  is $x$ and the probability of no contact in slot $t$ after a contact in slot $t-1$ is $y$.  The figure indicates that when $x$ varies between 0.3 and 0.7, a significant fraction of points is close to the line $x=y$.  This behavior is similar to the one expected in case contacts are approximately uniform and independent from each other.

\subsection{Evaluating Aging Control Policies}

\label{sec:evaluating}

In this section we validate our model against traces.   We begin by describing our methodology and reference configuration,  and then consider both location-oblivious and location-aware policies.

\subsubsection{Methodology and Reference Configuration}

To validate our model we use the traces described in the previous section.  We assume that users do not have a cell phone data plan ($P_{3G}\rightarrow \infty$) and   WiFi is  free  ($P=B=0$).

The computation of   the optimal policy using the proposed model requires estimates of  $p$, $U(x)$ and $G$. For a given bus shift, $p$ is estimated as the   number of useful slots in that shift divided by the total number of slots.   Note that to compute the optimal  strategy using our model we assume knowledge of $p$, but not of the  distribution of contact times and durations.   When searching for the optimal threshold strategy using  traces, in contrast, we perform trace-driven simulations. Our simulator, as well as  extensive statistics obtained from the traces, are available at \url{http://www-net.cs.umass.edu/~sadoc/agecontrol/}. For each bus shift, our simulator takes as input the string of ones and zeros corresponding to useful and unuseful slots, and computes the reward experienced by a user adopting a given activation policy.  The strategy  that yields the highest reward correspond to the \emph{optimal trace-driven policy}.

In our reference setting, the utility of messages decays linearly during a bus run,  and remains zero afterwards, $U(x) = \max(M-x,0)$. 
We vary  $M$ between 10 and 16 (which correspond to 50 and 80 minutes, resp., see also Figure~\ref{intercontact2}(a)), in increments of 2.   For ease of presentation, let $b$ be the energy cost scaled by a factor of  $1/(M-1)$, \eat{and present our results as a function of $b$,}
\begin{equation} \label{eq:bgm1}
b=G/(M-1)
\end{equation}
We  vary $b$ according to our experimental goals between 0.2, 0.8 and 1.8, corresponding to small, medium and high costs, respectively.

\begin{figure}
\center
\includegraphics[scale=0.45]{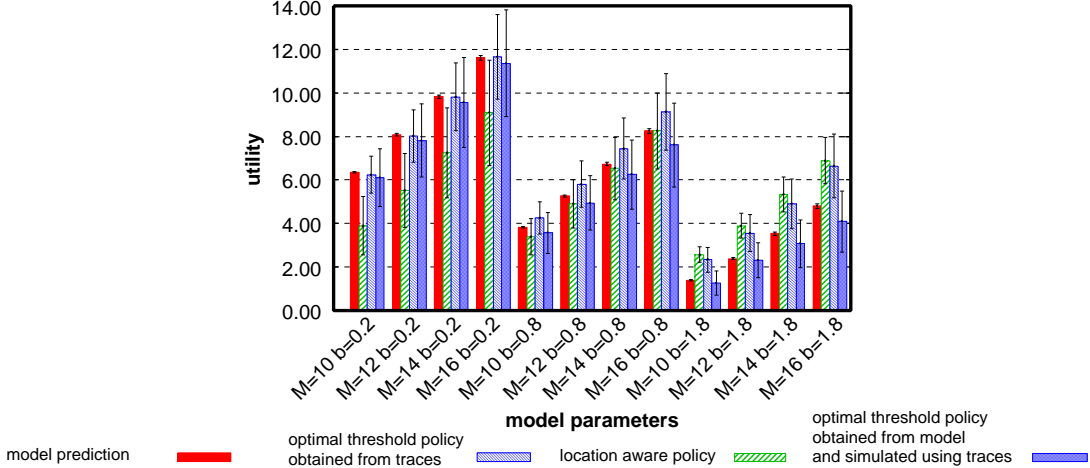}
\caption{Rewards as predicted by our model and observed in DieselNet.}
\label{fig:utility}
\end{figure}

\subsubsection{Location Oblivious Policies} Next, we validate our model against traces assuming that users are restricted to threshold policies (non-threshold policies are considered in the next section).  Figure~\ref{fig:utility} shows, for different model parameters, the averages of  1) the optimal  reward  \emph{predicted} by our model, 2) the reward \emph{effectively obtained} using the optimal strategy computed by the model as input to  our trace-driven simulator and 3)  the \emph{optimal trace-driven (threshold) policy}, with the 95\% confidence intervals (to increase the number of samples, each bus shift is replayed 40 times). \eat{\footnote}

In Figure~\ref{fig:utility}, note   that across all parameters, the rewards predicted by our model match the rewards effectively obtained  using the policy proposed by our model pretty well.  When the energy cost is low ($b=0.2$),  our model predictions also closely match the optimal trace-driven policy.    When the energy cost is medium ($b=0.8$) the accuracy of the predictions of our model depends on the maximum age $M$.  Recall that typical bus runs last between 40 minutes and 1 hour and 20 minutes (see Figure~\ref{intercontact2}(a)) and that when a bus run is completed, a contact occurs with high probability  (see Figure~\ref{history}(a)).  If the maximum age is 50 minutes ($M=10$), the fact that our model does not capture the correlations among contacts between buses and APs does not play an important role.   However, as $M$ increases,  strategically setting the policy to account for such correlations  is  relevant.   Similar reasoning holds when~$b=1.8$.

Figure~\ref{utility1a}(a) shows  the distribution of the distance between the optimal threshold using our model versus the  optimal trace-driven policy, for different model parameters.  In accordance to our previous observations, when $b=0.2$, the distance between the optimal threshold and the one computed by our model is smaller than or equal to two in at least 60\% of the bus shifts.  The same holds when $b=0.8$ and $M=10$.  When $M=16$ and $b=1.8$, in contrast, the distance to the optimal threshold is smaller than two for less than 40\% of the bus shifts.  

\begin{figure*}
\hspace{-0.2in}
\includegraphics[scale=0.32]{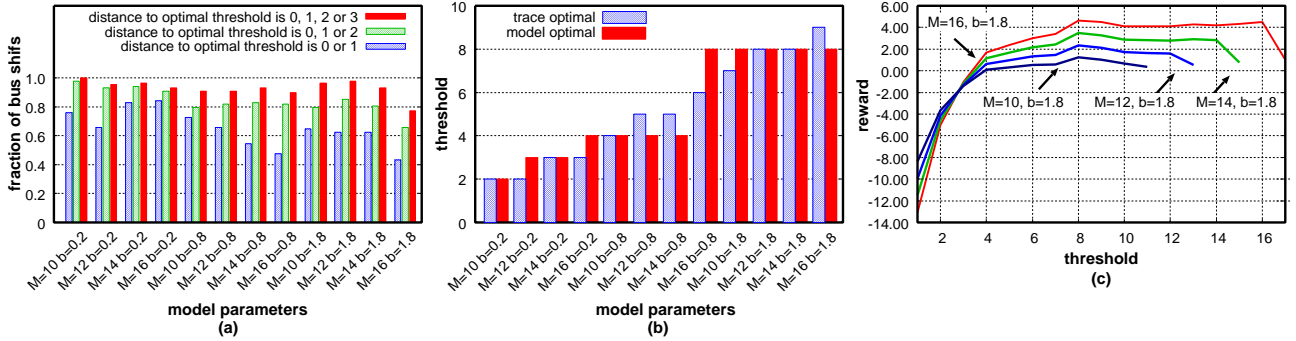}
\caption{Model validation (a) distance between optimal trace-driven  threshold and optimal model threshold; (b) optimal threshold mode; (c)  trace-driven reward as a function of the threshold, where threshold is assumed the same at  all bus shifts.}
\label{utility1a}
\end{figure*}

Figure~\ref{utility1a}(b) reports the mode of the optimal threshold as estimated by our trace-driven simulator and by our model.  The mode of the  optimal threshold increases with the energy cost (see Proposition~\ref{increasing}), and the distance between the simulator and model predictions never surpasses two.

Figure~\ref{utility1a}(c) provides further insight on the problem when $b=1.8$.   
Until this point, we allowed different bus shifts to correspond to different threshold policies (i.e., we considered bus shift discriminated strategies).  Next, we consider users that adopt the same threshold policy over all bus shifts (i.e., we consider  flat strategies over bus shifts).    For different threshold values, Figure~\ref{utility1a}(c) shows the average reward obtained using  our trace-driven simulations (see Figure~\ref{fig:utility} for confidence intervals).  Note that the optimal  threshold equals eight  in the four scenarios under consideration.  This threshold value, in turn, is in agreement with the mode predicted by our model (see Figure~\ref{utility1a}(b)).   In Figure~\ref{ploted}(a)  we reproduce the results of Figure~\ref{utility1a}(c) using our model.  We let $p=0.53$, the median contact probability (see Figure~\ref{intercontacta}(c)).   Comparing Figure~\ref{utility1a}(c)  and Figure~\ref{ploted}(a), we observe that the empirical curves are predicted by our model with remarkable accuracy.  The distance between the optimal threshold predicted by our model (marked with circles in Figure~\ref{ploted}(a)) and the one obtained through the trace-driven simulations (consistently equal to eight) is smaller than or equal to 1, and the utility discrepancy does not surpass~0.5.

\label{sec:refermodel}

\subsubsection{Location Aware Policy} Users can leverage geographic information in order to decide when to activate their mobile devices.  At UMass, for instance, a central bus stop is located at the  Haigis Mall (see Figure~\ref{history}(b)).  Buses usually linger at Haigis Mall for a couple of minutes, time at which transmission opportunities  usually  arise (see Figure~\ref{history}(a)). Therefore, in this section we consider the following location aware policy:  activate the mobile device if in Haigis Mall, and remain inactive otherwise.  Figure~\ref{fig:utility} shows the performance of our location aware policy obtained using trace-driven simulations.  When the energy  costs are not high ($b=0.2$ or $0.8$), it is always advantageous to  activate the mobile device before the bus returns to Haigis Mall, in order to opportunistically take advantage of contacts with APs during the bus run.  However, when the energy costs are high ($b=1.8$), our location aware policy outperforms the best trace-driven threshold policy.  The non-optimality of threshold policies in this scenario corroborates the fact that when energy costs are high, it is important to account for correlations among contact opportunities.

\subsubsection{Summary}

To sum up, our model accurately predicts the reward effectively obtained by its proposed policy.  If the energy costs are low, or if users do not discriminate their strategies with respect to the bus shifts, our  model can also predict   the optimal policy obtained through trace-driven simulations.  Otherwise, users can benefit from correlations among contact opportunities.

\label{locationaware}

\section{Model Analysis}

Our goals now are $(a)$ to  derive the optimality conditions that must be satisfied by the optimal policy; $(b)$  to show properties of the  optimal threshold and $(c)$  to present specialized results for step and linear utility functions.  We tackle each of the goals in one of the subsequent sections, respectively.

\label{optimal}

\subsection{Optimal Policy General Structure}

\label{sec:genstr} 

We now derive the general structure of the optimal policy.  To this goal, consider a fixed policy $u$. 
In this section we assume that, under $u$, users have a positive activation probability in at least one state.  The conditions for the optimality of the policy which consists of remaining always inactive will be established in Proposition~\ref{propoalin}.  

 Let $\mathcal{P}_{u}$ denote the transition probability matrix of the Markov chain $\{ x_t : t=1,2, \ldots \}$ which characterizes the dynamics of the  age,  given policy $u$.  Let $r_u(x)$ be the expected instantaneous reward received in a time slot when the system is in state $x$ and policy $u$ is used.  The vector of expected instantaneous rewards is denoted by $\mathbf{r}_u$.   
 
 Let the gain, $g_{u}$, be the average reward per time slot in steady state, $g_{u}=E[r; u]$. Since the number of states in the system is finite and from each state there is a positive probability of returning to state~1, $\mathcal{P}_{u}$  comprises a single connected component, and $g_{u}$ does not depend on the initial system state.

\begin{table}
\center
\begin{tabular}{lll}
\hline
variable & description & \\
\hline
$a_{t}$ & action of user \\
$x_{t}$ & age of message \\
$r_{t}$ & reward  \\
$e_{t}$ & \multicolumn{2}{l}{=1 in case of useful contact opportunity, =0 otherwise} \\
\hline
parameter & description & determined by \\
\hline
$\utility(x)$ & utility of message at age $x$ & user \\
$B_t$ & bonus level & publisher \\
$P_{3G}$ & price of 3G &  3G service provider \\
$P$ & price of WiFi & WiFi service provider \\
\hline
$G$ & activation cost \\
$p$ & \multicolumn{2}{l}{probability of useful contact  opportunity, $p=E[e_{t}]$} \\
\hline
\hline
\end{tabular}
\caption{Table of notation.  Notes: (1) Subscripts are dropped when in  steady~state. (2) $B$ is a parameter in~\textsection\ref{optimal} and a variable in~\textsection\ref{maker}. 
}
\label{notation}
\end{table}

The \emph{relative reward} of state $x$ at time $t$, $V(x,t)$,  is the difference between the expected total reward accumulated when the system starts in $x$ and the  expected total reward accumulated when the system starts in steady state.  Let  $V(x)=\lim_{\ell \rightarrow\infty} V(x,\ell) =\lim_{\ell \rightarrow\infty} \Big( \sum_{t=0}^{\ell} \mathcal{P}_{u}^t (\mathbf{r}_u -g_{u} \be) \Big)(x)$, 
  where $\be$ is a column vector with all its elements equal to one.  It follows from ~\eqref{rtiddef} and similar arguments as to those in~\cite[eq. (8.4.2)]{puterman} that a policy which satisfies  the following conditions,  for $1 \leq x \leq M$, is optimal,
\begin{eqnarray}  \nonumber
V(x)= \max \Big(\utility(x) +  V(\min(x+1,{M})) -g_{u},  \utility(x) {-}  \cost  {+}  p (V(1){-}P{+}B)  {+}  (1-p) V(\min(x+1,{M})) {-}g_{u}\Big) \nonumber \end{eqnarray}
An optimal policy is obtained from the optimality conditions as follows,
\begin{equation}   
a(x) {=} \left\{ \begin{array}{ll}  {0}, & \textrm{ if }    {{-} {\cost}/ {p}  {+}   {(}V(1){-}P{+}B } {)}    {\leq}  V(\min(x+1,{M}))  \\ \label{eq3}
{1},  & \textrm{ otherwise} \end{array} \right. 
\end{equation}
In~Appendix~\ref{sec:no3g}  we show that $V(x)$ is decreasing on $x$ when $P=0$ (the corresponding result when $P \ge 0$ can be found in~\myapp{app:vdecr}).  Thus, the existence of an  optimal policy in the class of threshold policies (Proposition~\ref{temp1}) follows from ~\eqref{eq3}. 

Note that adding a constant to  $U(x)$, $1 \leq x \leq M$, does not affect the optimal policy ~\eqref{eq3}.  Therefore, in the rest of this paper we assume, without loss of generality, that $U(M)=0$.

\subsection{Optimal Threshold Properties}

\label{sec:aver}  
  
In this section we aim at finding properties of the optimal threshold.  To this goal, we note that an user adopting a threshold strategy goes through cycles.  Each cycle consists of an  idle and active period.  An idle period is initiated when the age is one, and ends immediately before the instant at which the age reaches the threshold $s$.  At age $s$, an active period begins and lasts on average $1/p$ time slots, up to the instant at which the age is reset to one.  

The following proposition establishes conditions according to which the optimal  actions are invariant in time (see Appendix~\myapp{sec:no3g}).   The proof of the subsequent results \eat{when not found in the appendix,} are available in the appendices. \eat{at~\myapp{sec:no3g}.}

\begin{proposition} \label{propoalac}
The optimal policy consists of being always active if and only if
\vspace{-0.05in} 
\begin{equation}
\frac{1}{1-p} \left( \utility(1)-p \sum_{x=1}^{M-1} \utility(x) (1-p)^{x-1}  \right)  \geq   \frac{\cost}{p} +P-B  \label{cond:alwaysactive}
\end{equation}
\label{propoalin}
The optimal policy consists of being always inactive if and only if 
\vspace{-0.1in}
\begin{eqnarray}
\sum_{j=1}^{M-1}  \utility(M-j)  &\leq& \frac{\cost}{ p} + P-B \label{cond:alwaysinactive}
\end{eqnarray}

\end{proposition}

Given that the condition for the optimal policy to be always inactive is  established in the second part of Proposition~\ref{propoalin}, in what follows we focus on policies which  consist of being active in at least one state.

For a fixed threshold policy with corresponding threshold $s$, let the system state transition probability matrix, $\mathcal{P}_{s}$, be

{
\footnotesize
\begin{equation}
\mathcal{P}_{s}=
\begin{matrix}
1
\\
\\
\\
s
\\
\\
\\
M
\end{matrix}
\begin{bmatrix}
0  &  1 &  &        &   & \\
0  &  0 & 1&        &   & \\
   &    &  & \ddots &   & \\
   p  &    &  &        & 0 & 1-p & \\
   &    &  &   & & \ddots &   & \\
p  &   &   &     & &   &      &  1-p \\
 \end{bmatrix}
\end{equation}
}

Let $\bpi$ be the steady state solution of the system, $\bpi \mathcal{P}_{s}  =  \bpi$.   Then, the fraction of time slots in which an user issues updates is  $\pi_1$ (see Appendix~\ref{sec:exavgre}), 
\vspace{-0.15in}
\begin{eqnarray} \label{exavgre}
\pi_1&=&\left( {s+\frac{1-p}{p} } \right)^{-1}
\end{eqnarray}

Next, we derive a closed-form expression for the expected reward as a function of the threshold $s$.    In~\myapp{app:ers}  we show that replacing the expression of the steady state  solution $\bpi$ into $E[r;s]= \sum_{i=1}^{M} \pi_i r(i,\mathbf{1}_{i\ge s})$ yields,
\begin{equation} \label{eq:ers} 
\boxed{
E[r;s] {=}  {\pi_1} \left[ \sum_{x=1}^{s-1} \utility(x)  {+} \sum_{i=0}^{M-1-s} \utility(i+s) (1-p)^{i}  {-} \frac{\cost}{p} {-}P{+}B \right] } \nonumber
\end{equation}
The derivation of the equation above is found in Appendix~\ref{sec:erseq}. 
In~\myapp{app:propos1b} we use the expression above  to show that the expected average reward  is non-decreasing in $s$, for $s < s^{\star}$, as stated in the following proposition.
\begin{proposition} \label{propos1b}
The optimal threshold value $s^{\star}$ is
\begin{equation} \label{def:1}
s^{\star} = \min\left\{ s \Big| E[r; s] \ge E[r; s+1] \right\}
\end{equation}
\end{proposition}

\begin{figure}
\center
\includegraphics[scale=0.42]{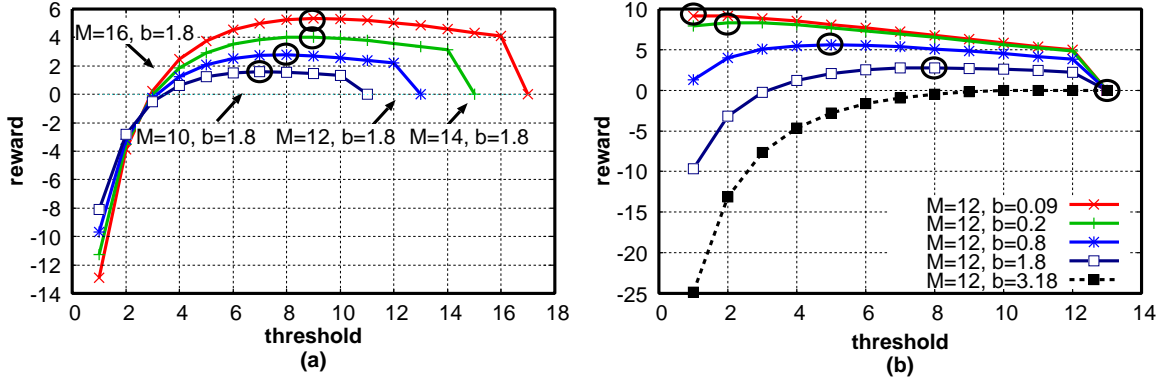}
\caption{Model numerical evaluation ($U(x)=M-x$, $P=0$, $p=0.54$).}
\label{ploted}
\end{figure}

Finally, Proposition~\ref{increasing} formalizes the monotonicity of $s^{\star}$ with respect to $G$  observed in~\textsection\ref{eval}.
\begin{proposition}  The optimal threshold $s^{\star}$ is non-decreasing with respect to  $G$ and $P$ and non-increasing with respect to~$B$. \label{increasing} \end{proposition} 

\vspace{-0.25in}

\subsection{Special Utility Functions}  \label{sec:special}

\vspace{-0.05in}

Next, we specialize our results to two classes of utilities.

\subsubsection{Step Utility} Let the utility function $\utility(x)$ be given by a step function, $\utility(x)= v$ if  $x \leq k$ and $U(x)=0$ otherwise. 
The next proposition allows us to efficiently compute the optimal threshold in this case.   Let $\varphi$ be the root of $\frac{d}{ds} E[r;s]=0$ when $k \leq  s$.  The  closed form expression of $\varphi$ is given in~\myapp{app:step}. 
\begin{proposition} \label{propo:step}
For the step function utility, the optimal threshold $s^{\star}$ is equal to either 1, $k-1$, $\lfloor \varphi \rfloor$, $\lceil  \varphi \rceil$, $M$ or $M+1$.
\end{proposition}

\subsubsection{Linear Utility}

Let $\utility(x)=M-x$. Then, $E[r;s]$ can be expressed in closed form as a function of $s$, $M$, $p$ and $G$ (see~\myapp{app:linear}).    To illustrate the behavior of $E[r;s]$, we let $p=0.54$ (see Figure~\ref{intercontactb}), $P=B=0$, and vary $b$ and $M$ as shown in the legend of Figure~\ref{ploted}.   Figure~\ref{ploted}(a) was discussed in~\textsection\ref{sec:refermodel} in light of our trace-driven simulations.     Figure~\ref{ploted}(b)   shows that when    $b=0.09$ (resp., $b=3.18$),  inequality ~\eqref{cond:alwaysactive} (resp., ~\eqref{cond:alwaysinactive}) holds and the optimal threshold is 1 (resp., 13), in agreement to Proposition~\ref{propoalac}.  In accordance to Proposition~\ref{propos1b}, Figure~\ref{ploted}(b)  also indicates that when $s < s^{\star}$ the reward is increasing.  Finally, the optimal threshold in Figure~\ref{ploted}(b)    increase as a function of $b=G/(M-1)$, which serves to illustrate Proposition~\ref{increasing}.

\section{The  Publisher Bonus Package}

In this section we  consider strategic publishers that offer bonus packages to users, so as  to incent them to download updates  of  advertisement campaigns or unpopular content.    In~\textsection\ref{complete} we assume  that publishers have complete information on the system parameters and consider the incomplete information case in~\textsection\ref{incomplete}.

\label{maker}

\subsection{Complete Information}

\label{complete}

Next, we consider publishers that, while devising their optimal bonus strategies,  leverage the fact that users solve the {\ensuremath{\mbox{\sc  User Problem}}}.  The optimal bonus strategy consists of finding  the bonus level $B$ that minimizes the average age of messages in the network,  under the constraint that the expected number of messages  transmitted  per time slot is below a given budget, dictated by the service provider.

Let $N$ be the number of users in the network. In what follows, we make the dependence of $s^{\star}$ and $\pi_1$ (see~\eqref{exavgre}) on the bonus level $B$ explicit.  Let $Q$ be the average number of messages transmitted per time slot.
\begin{equation}
Q= N  \pi_1(B) = N  / (s^{\star}(B)+(1-p)/p)
\end{equation}
Let $A$ be the average age of messages in the network (see~\myapp{app:expectedage} for its closed-form expression). Let $T$ be the constraint imposed by the service provider on the expected number of messages sent per time slot.  Then, the problem faced by the publisher~is

\medskip
\boxed{
\begin{minipage}{6.3in}
{\ensuremath{\mbox{\sc  Publisher Problem:}}} Obtain bonus level $B$ so as to
\begin{eqnarray}
\min && A=\sum_{i=1}^{M} i \pi_i(B)  \label{objectivepublisher} \\
s.t. && Q \leq T \label{ineq2}
\end{eqnarray}
\end{minipage}
}

Note that since $s^{\star}(B)$ is not injective, there might be a range of bonuses levels that solves the {\ensuremath{\mbox{\sc  Publisher Problem}}}.   Under the assumption that the above problem admits a solution, which is guaranteed to exist if $s(P) \ge {N}/{T}-({1-p})/{p}$, 
\begin{proposition}  \label{optimalpubprob}
The solution of the {\ensuremath{\mbox{\sc  Publisher Problem}}} consists of setting the bonus level $B$ in such a way that  $s=\max(0, \min(M+1, \lceil N/T - (1-p)/p \rceil ))$.  The solution can be found using a binary search algorithm.
\end{proposition}

\subsection{Incomplete Information: Online Learning Algorithm}

\label{incomplete}
How publishers can set their bonus levels without knowing the number of users in the system and their  	strategies?  To answer this question, we present a simple learning algorithm to solve the {\ensuremath{\mbox{\sc  Publisher Problem}}} when the system parameters are unknown.

\begin{algorithm}
\caption{Online estimation of optimal bonus level.}
\label{alg:algorit}
\begin{algorithmic}[1]
\begin{small}
\STATE \textbf{Input: }  maximum bonus level $\widehat{B}$, target number of messages per slot $T$, round duration $\tau$ time slots, learning rate $\alpha$ 
\STATE Choose initial bonus level $B_0$ such that $B_0\in [0,\widehat{B}]$; $t \leftarrow 0$ \WHILE {$|T - Q_t| > \epsilon$}
\STATE \emph{At the end of round $t$,}
\STATE $Q_t \leftarrow R_t/\tau$ 
\STATE $ B_{t+1} \leftarrow \min(\widehat{B},\max(0, \left(B_t+  \alpha\left( T -  Q_t \right)/t\right)))$
\STATE $t \leftarrow t+1$
\ENDWHILE
\end{small}
\end{algorithmic}
\end{algorithm}

\begin{figure}
\center
\includegraphics[scale=0.41]{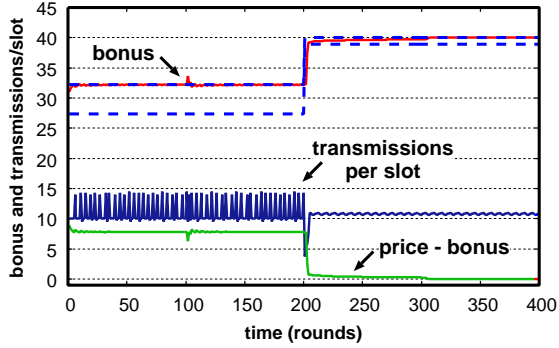}
\caption{Learning algorithm: sample path of trace-driven simulation} \label{fig:learning}
\end{figure}

The proposed algorithm proceeds in rounds.  Each round corresponds to $\tau$ time slots,  at which users have their requests served.  Let $B_t$ be the bonus set by the publisher at the beginning of round $t$, and let  $R_t$ be the number of requests served at that round.  The average number of requests served per time slot in round $t$ is  $Q_t = R_t/\tau$.

Algorithm~\ref{alg:algorit} updates the bonus level as follows.  At round $t$, if the average number of requests served  per time slot is above the target $T$, the bonus level is decreased by $\alpha (Q_t-T)/t$; otherwise, the bonus level is increased by $\alpha (T-Q_t)/t$.
The learning step size,  $\alpha$, is a learning parameter that impacts the algorithm convergence time. Smaller values of $\alpha$ yield a smoother but slower convergence~\cite[Chapter 5]{kushner}.   
The bonus level is required to be positive,  and cannot surpass the maximum bonus level, $\widehat{B}$ (line 6 in Algorithm~\ref{alg:algorit}). 
The algorithm stops when $|T-Q_t| \leq \epsilon$, where $\epsilon >0$ is the  tolerance parameter~(line~3).   Under the assumption that the system parameters (see Table~\ref{notation}) are fixed,   
\begin{proposition} \label{proposition:converge} \vspace{-0.05in}
The sequence of bonus levels $\{ B_t \}_{t=0}^{\infty}$ converges to the optimal solution of the {{\ensuremath{\mbox{\sc  Publisher Problem}}} }  with probability one. \end{proposition}

\vspace{-0.05in}

According to Proposition~\ref{proposition:converge}, the convergence of Algorithm~\ref{alg:algorit} is assured when the system parameters are fixed.    In order to 1) study the behavior of our algorithm when the population size varies and to 2) investigate how the convergence speed may be affected by the correlation among users, we conducted trace-driven simulations, whose results  are   reported in~\myapp{app:simulation}.  Next, we use Figure~\ref{fig:learning} to illustrate some of our findings.

Figure~\ref{fig:learning} shows a sample path of our  trace-driven simulations.  We let $\alpha=1$, $M=30$, which corresponds to 2:30 hours (see Figure~\ref{intercontacta}(a)),  $\tau=100$,  which corresponds to up to  3  updates per day, $p=0.54$ (see Figure~\ref{intercontacta}(b)), $G=0.4$, $P=\widehat{B}=40$ and $T=11$.   The number of users is initially 50, and decreases to 20 at round 200.   We assume that users solve the  {\ensuremath{\mbox{\sc  User Problem}}}   when setting their activation strategies, and     Algorithm~\ref{alg:algorit} is run by the publisher every 100 rounds.  We  consider half of the population in one bus and the other half in another (see~\myapp{app:simulation}).  Despite the correlations among users,  Figure~\ref{fig:learning} does not qualitatively deviate from the results obtained with  uncorrelated users.  In particular, the algorithm converges in up to 20 rounds, and the number of transmissions per slot varies between 9 and 11. The oscillations are not always centered at $T=11$ due to the fact that the  threshold adopted by the users is integer valued, which might prevent inequality~\eqref{ineq2} to~bind.     Note that the bonus level converges to values in the optimal range, obtained using Proposition \ref{optimalpubprob} and marked with dotted lines in Figure~\ref{fig:learning}.  When $N=20$, the publisher  sponsors the service provider costs ($B=P$), and users make their updates for~free.

\vspace{-0.05in}
\section{Discussion of Assumptions, Limitations and Future Work}

Next, we discuss the main simplifications adopted to yield a tractable model (Assumption~\ref{assum:uniform} is discussed in details in ~\textsection\ref{eval}, so we do not include it in the following list).

\label{limitations}

\textbf{Frequent decisions assumption: }   We assume that users are interested in maximizing their expected average rewards. This assumption is appropriate if decisions are made frequently.  In our  measurement study,  we assumed that users make decisions every 5 minutes, and we observed that for a vast number of bus shifts the system parameters are stable over hours.  

\textbf{High load assumption: }  We assume that updates are issued every time slot with high probability. As discussed in~\textsection\ref{sec:workload}, this assumption holds for certain news feeds categories such as \emph{top stories}.  
The high load assumption also holds if users synchronize the updates of multiple applications, using APIs such as OVI Notifications~\cite{ovi}.  In this case, the larger the number of applications that require updates, the higher the chances of at least one update being available every time slot.  If multiple application updates can be performed at a single time slot, the bundle being subject to no additional costs or bonus, our model holds without modifications.  However, if the interaction among the applications is non trivial, new policies need to be devised accounting for decisions such as when  and from whom to download updates from each of the applications (see~\textsection\ref{related} for complementary work on that topic).

\textbf{Self-regarding users assumption: } We assume that users do not collaborate with each other.  Although researchers are interested in leveraging collaboration among users~\cite{pocket,giovanni},  a large number of mobile systems still does not take advantage of peer-to-peer transfers, and users need to download their messages exclusively from access points or base stations~\cite{podcast}.    Our model applies to such systems. Nevertheless, we envision that if peers are roughly uniformly distributed,  peer-assisted opportunistic transfers can be easily captured by our model simply by adapting the contact probability, $p$, in order to  account for peer-to-peer contacts.  Future work consists of  accounting for spatial information when modeling the likelihood of contact opportunities.  

Finally, note that in this paper we do not consider  prediction algorithms in order to infer future opportunities with WiFi access points.   In particular, some of  the algorithms without  prediction previously proposed in the literature, such as \cite[\textsection 6.1.1]{wifler}, are special instances of the aging control policies described in this paper.     As shown in~\textsection\ref{eval}, neglecting the possibility of prediction comes with no significant loss in scenarios of practical interest,  the incorporation of prediction algorithms in our framework being subject for future work. 

\vspace{-0.05in}
\section{Related Work}
\label{related}

The literature on  measuring~\cite{Thedu-mobicom08,Zhang:2007,podcast}, modeling~\cite{chaintreau, pocket, chants, qest, meshes} and control~\cite{ altman1,activation, impatience,giovanni,battery,wu1,wu2} in wireless networks is vast.   Nevertheless, we were not able to find any previous study on the \emph{aging control}  problem as described in this paper.  Previous work accounted for the modeling of aging~\cite{chaintreau} or for the age control by publishers~\cite{altman1, impatience}, but not for users aging control as described in this paper.    We were also not able to find any previous study on the the implications of \emph{bonus packages} set by content and/or service  providers~\cite{sfr}.

Chaintreau \emph{et al.}~\cite{chaintreau} model the distribution of message ages in a large scale mobile network using a spatial mean field approach.  Their model allows the analysis of gossiping through opportunistic contacts. 
In this paper, in contrast, we assume that nodes rely exclusively on base stations and access points in order to receive their updates. 

Activation control strategies were proposed in~\cite{infocom2010, altman1, activation, link}.  In~\cite{altman1} the authors consider publishers of evolving files,  that aim at reducing their  energy expenditure by controlling the probability of transmitting messages to  users. In \cite{activation}, a joint activation and transmission control 
policy is proposed so as to maximize the throughput of  users under energy constraints.   In~\cite{link}, a joint activation and link selection control policy is proposed so as to minimize the energy consumption under delay constraints. Our work differs from~\cite{altman1, activation, link} in two ways as 1) we investigate the activation policy of mobile nodes based on the utility of their  messages and 2) we study the publishers bonus strategy. \eat{,   3)  our analysis is  carried out using  MDPs and 4) we validate our conclusions using traces from UMass DieselNet.  }  In addition, our analysis is carried out using MDPs and is validated using UMass Dieselnet traces (similar methodology is adopted, for instance, by Yang \emph{et al.}~\cite{optimalyang}).

The utility function introduced in this paper corresponds to the impatience function presented by Reich and Chaintreau~\cite{impatience}. 
Reich and Chaintreau~\cite{impatience}  study the implications of delays between requests and services on optimal content replication schemes.      
If users have limited caches and cannot download all the requested files every time they are in range of an access point, 
the insights provided by~\cite{impatience, ioannidis2} need to be coupled with the ones presented here in order to devise the optimal joint activation-replication~strategy.  Therefore, aging control, as described in this paper, significantly complements replication control, as described in~\cite{impatience, ioannidis2}. \eat{, are complementary to each~other.}

\vspace{-0.05in}

\section{Conclusion}

\label{conclusion}

This paper reports our measurement and modeling studies of aging control in hybrid wireless networks.  From the DieselNet measurements, we learned that correlations among contact opportunities do not play a key role if the energy costs incurred by users are small or if users cannot discriminate their strategies based on  bus shifts.  We then modeled and solved the aging control problem, and used trace-driven simulations to show that a very simple threshold strategy derived from our model performs pretty well in practice.  When publishers are strategic, we analyzed the bonus package selection problem, and showed that in some scenarios it is beneficial for publishers to fully sponsor the content updates  requested by users.  We believe that the study of mechanisms to  support   applications that require frequent updates, such as microblogging, in wireless networks, is an interesting field of research, and we see our paper as a first attempt to shed light into the tradeoffs faced by users and publishers of such applications. 

\vspace{-0.08in}

{
\footnotesize
\bibliographystyle{IEEEtran}
\bibliography{dtn1}
}

\clearpage
\pagebreak

\begin{appendix}

\section{Derivation of Main Results When $P_{3G} \rightarrow \infty$}
\label{sec:no3g}

We begin analyzing the case in which 3G is not available.  The optimality conditions are
\begin{eqnarray}   \boxed{\boxed{
V(x)+g_u{=} \max \Big(\utility(x) {+}  V(\min(x+1,{M})),  \utility(x) {-}  \cost  {+}  p (V(1){-}P{+}B)  {+}  (1-p) V(\min(x+1,{M})) \Big), \quad x{=}1, \ldots, M  \label{eq2}   }}
 \end{eqnarray}
Let
\begin{eqnarray}
H(x,0)&=&U(x) + V(\min(x+1,{M})) \label{eq2A} \\
H(x,1)&=&  \utility(x) {-}  \cost  {+}  p (V(1){-}P{+}B)  {+}  (1-p) V(\min(x+1,{M})) \label{eq2B}
\end{eqnarray}

It follows from \eqref{eq2}-\eqref{eq2B} and ~\cite{puterman} that the following policy is optimal, for $m=1, \ldots, M$,
\begin{eqnarray} \label{eqeq1}
a(m)= \left\{ \begin{array}{ll}  0, & \textrm{if } H(m,0) \ge H(m,1)  \\
1, & \textrm{otherwise}    \end{array} \right.
\end{eqnarray}

At states $M$ and $M-1$, \eqref{eq2}  implies that
\begin{eqnarray}  \label{eqeq1}
 \left\{ \begin{array}{ll}  (H(M,0) \ge H(M,1)) \wedge (H(M-1,0) \ge H(M-1,1)), & \textrm{ if }    V(M) \ge  \utility(M) {-}  \cost  {+}  p (V(1){-}P{+}B)  {+}  (1-p) V(M) \\
(H(M,0) \le H(M,1)) \wedge (H(M-1,0) \le H(M-1,1)), & \textrm{otherwise}
\end{array} \right.
\end{eqnarray}
Equation \eqref{eqeq1} yields the following remark,  used in the analysis of the  base cases of the  inductive arguments that follow.
\begin{remark} $(H(M,0) \ge H(M,1))  \iff (H(M-1,0) \ge H(M-1,1))$.  \label{remark1}
\end{remark}

\subsection{Derivation of ~\eqref{exavgre} and corresponding steady state probabilities} \label{sec:pi1app}
\label{sec:exavgre}

From~\textsection\ref{sec:aver},
\begin{eqnarray} 
\pi_1 &=& \pi_i, \quad i=1, \ldots, s  \label{eq1a} \\
\pi_{i} &=& \pi_{i-1} (1-p)=\pi_1 (1-p)^{i-s}, \quad  i=s+1, \ldots, M-1\label{eq2a} \\
\pi_M &=& (1-p)(\pi_{M-1} + \pi_M) = \frac{1-p}{p}\pi_{M-1}
\end{eqnarray}
Therefore, 
\begin{eqnarray}
\pi_1&=&\left( {s+\sum_{i=s+1}^{M-1}  (1-p)^{i-s}+(1-p)^{M-s} /p} \right)^{-1} = \left( {s+\sum_{i=1}^{M-1-s}  (1-p)^{i}+(1-p)^{M-s} /p} \right)^{-1}  = \left( {s+\frac{1-p}{p} }  \right)^{-1} \nonumber
\end{eqnarray}

\subsection{Derivation of expression of $E[r;s]$} \label{sec:erseq}
Next, we show that 
\begin{equation} \label{eq:ers}
E[r;s] {=}  {\pi_1} \left[ \sum_{x=1}^{s-1} \utility(x)  {+} \sum_{i=0}^{M-1-s} \utility(i+s) (1-p)^{i}  {-} \frac{\cost}{p} {-}P{+}B \right]
\end{equation}
\label{app:ers}
Let 
\begin{equation}
\gamma=\left( (U(M){-}G{-}pP{+}pB)\pi_M - \sum_{i=s}^{M-1}(G{+}pP{-}pB)\pi_i \right)/\pi_1
\end{equation}
Then,
\begin{eqnarray} \label{eq:erapp}
E[r;s]&=&\sum_{i=1}^{s-1} \utility(i) \pi_1 + \sum_{i=s}^{M} (\utility(i)-\cost -pP+pB) \pi_i =  \pi_1 \left[ \sum_{i=1}^{s-1} \utility(i)  + \sum_{i=s}^{M-1} \utility(i)\pi_i/\pi_1   + \gamma \right] \\
&=&  \pi_1 \left[ \sum_{i=1}^{s-1} \utility(i)  + \sum_{i=0}^{M-1-s} \utility(i+s) (1-p)^{i}  + \frac{1}{p} (1{-}p)^{M-s} (\utility(M){-}\cost{-}pP{+}pB){-}(\cost{+}pP{-}pB) \frac{1{-}(1{-}p)^{M{-}s}}{p}\right]  \nonumber \\
&=&  \pi_1 \left[ \sum_{i=1}^{s-1} \utility(i)  + \sum_{i=0}^{M-1-s} \utility(i+s) (1-p)^{i}   -G/p - P + B\right]  \label{eqersspecial} 
\end{eqnarray}

\subsection{Proof of Proposition~\ref{vdecr_no3g}}
\begin{proposition} \label{vdecr_no3g} 
The {\ensuremath{\mbox{\sc  User Problem}}} admits an optimal threshold policy 
\begin{eqnarray}
{a(x)} {=} \left\{ \begin{array}{ll}  0, & \textrm{ if }   x < s^{\star} \\
1,  & \textrm{ if } s^{\star} \leq x \leq  M.
\end{array} \right.
\end{eqnarray}
\end{proposition}

\label{app:vdecr}

\begin{proof}
We show that, for $m=0, \ldots, M-1$, 
\begin{itemize}
\item \begin{equation} \label{eqdecreasingv}
V(M-m-1) - V(M-m) \ge 0
\end{equation}
\item \begin{equation} 
( H(M-m,0) \ge H(M-m,1) ) \Rightarrow  ( H(M-m-1,0) \ge H(M-m-1,1) ) \label{existencecond}
\end{equation} 
\end{itemize}
We consider two scenarios, $H(M,0) \ge H(M,1)$ and~$H(M,0) < H(M,1)$.

\begin{center}
\begin{equation}
\boxed{ \boxed{
\textrm{{scenario 1)  } } H(M,0) \ge H(M,1)} }
\end{equation}
\end{center}

If $H$ satisfies~\eqref{existencecond} then, from~\eqref{eq2}-\eqref{eq2B},
\begin{equation}
V(m)=U(m)+V(m+1), 1 \leq m \leq M-1
\end{equation}
and ~\eqref{eqdecreasingv} holds. 
 
Next, assume for the sake of contradiction that $H$ does not satisfy~\eqref{existencecond}.  Let $m$ be the largest state at which condition~\eqref{existencecond} is violated,
\begin{equation}
m= \max\{ i |  H(M-i,0) \ge H(M-i,1) \textrm{ and } H(M-i-1,1) > H(M-i-1,0)  \} \label{minm}
\end{equation}
It follows from~\eqref{minm} that 
\begin{eqnarray}
H(M-m-1,1)&>& H(M-m-1,0) \label{minm1} \\
H(M-m+k,0)&\ge& H(M-m+k,1),  k = 0,1, \ldots, m  \label{minm2}
\end{eqnarray}
\eqref{minm1} and \eqref{minm2} yield, respectively,
\begin{eqnarray}
V(M-m) - V(1) &<& -G/p - P + B \label{yield1} \\
\eat{ V(M-m+1) - V(1) &\ge& -G/p - P + B \\ }
V(M-k) - V(1) &\ge& -G/p - P + B,  \qquad k = 0, \ldots, m-1 \label{hmkv1k0m1} 
\end{eqnarray}
Letting $k=m-1$ in~\eqref{hmkv1k0m1},
\begin{equation}
V(M-m) = U(M-m) + V(M-m+1) \ge U(M-m) + V(1) - G/p - P + B \label{yield2} 
\end{equation}

~\eqref{yield1} and ~\eqref{yield2}  yield the following contradiction
\begin{equation}
V(1) - G/p - P + B \le V(M-m) < V(1)-G/p - P +B
\end{equation}
Therefore, ~\eqref{existencecond} holds for $m=0, \ldots, M-1$.

\eat{
From the above, we conclude that $H(M-m-1,0)=H(M-m-1,1)$ which contradicts the fact that $H$ violates~\eqref{existencecond} at state $M-m-1$.} \eat{.  Let $a'(M-k)=a(M-k)$, for $k \leq m$ and $a'(M-m-1)=0$. If $a'$  satisfies ~\eqref{eqdecreasingv}-\eqref{existencecond}, the proof is complete. Otherwise, repeated application of the above procedure yields an optimal policy which satisfies~\eqref{eqdecreasingv}-\eqref{existencecond}. }

\begin{center}
\begin{equation}
\boxed{ \boxed{
\textrm{{scenario 2)  } } H(M,0) < H(M,1)} }
\end{equation}
\end{center}

\textbf{Base case: } We first show that $V(M-1) \ge V(M) \label{ineq1}$.

 Note that $(H(M,1) \ge H(M,0)) \Rightarrow (H(M-1,1) \ge H(M-1,0))$ (see remark~\ref{remark1}). It follows from  ~\eqref{eq2} that
 \begin{eqnarray}
 V(M-1)&=&\utility(M-1)-G-P+B+pV(1) +   (1-p)V(M)-g_{u} \label{base1b} \\
 V(M) &=&\utility(M)-G-P+B+pV(1)+  (1-p)V(M)-g_{u} \label{base2b}
 \end{eqnarray}
  Hence, \eqref{base1b}-\eqref{base2b}, together with the fact that $U(x)$ is non-increasing, yield $V(M-1) \ge V(M)$.

\textbf{Induction hypothesis [assume result holds for $m < t$]: } Assume that $V(M-m-1) - V(M-m) \ge 0$, for $m < t$, and  $(H(M-m,0) \ge H(M-m,1)) \Rightarrow (H(M-m-1,0) \ge H(M-m-1,1))$, for $m < t$.

\textbf{Induction step [show result holds for $m = t$]: } Next, we show that $V(M-t-1) - V(M-t) \ge 0$ and that 
$(H(M-t,0) \ge H(M-t,1)) \Rightarrow (H(M-t-1,0) \ge H(M-t-1,1))$.

It follows from the induction hypothesis that $V(M-t+1) \leq V(M-t)$.   We consider two cases,

\underline{$i)$ $H(M-t,0)\ge H(M-t,1)$}
The proof is similar to that of scenario 1.

\underline{$ii)$ $H(M-t,1) \ge H(M-t,0)$}   Now, we show that 
$$V(M-t-1) - V(M-t) \ge 0.$$
Note that $(H(M-t,0) \ge H(M-t,1)) \Rightarrow (H(M-t-1,0) \ge H(M-t-1,1))$ holds vacuously.    If $(H(M-t,1) \ge H(M-t,0))$   and $(H(M-t-1,1) \ge H(M-t-1,0))$,
\begin{eqnarray}
V(M-t) + g_{u} &=& \utility(M{-}t) {-}  \cost -P+B +  p V(1) +
 (1-p) V(\min(M-t+1,{M}))  \label{eq3aa} \\
V(M-t-1)+ g_{u}&=&
  \utility(M{-}t{-}1) {-}  \cost -P +B +   p V(1) +  (1-p) V(\min(M-t,{M})) \label{eq3ba}
  \end{eqnarray}
Also, $V(M-t+1) \leq V(M-t)$ (induction hypothesis) and $\utility(M-t) < \utility(M-t-1)$ (by assumption). Hence, \eqref{eq3aa}-\eqref{eq3ba} yield   $V(M-t) \leq V(M-t-1)$.

\end{proof}

\begin{remark} In what follows our analysis is restricted to threshold policies.
\end{remark}

\subsection{On the number of optimal thresholds}
\label{sec:appdiscussion}

The existence of policies satisfying~\eqref{eq2} follows from~\cite{cavazos1} and~\cite{cavazos2}.  In particular, the optimal threshold policies  satisfy ~\eqref{eq2}.    In general, though, the solution to ~\eqref{eq2} is not unique.  Next, we characterize the scenarios in which there are two or more optimal threshold policies.

Figure~\ref{fig:multiple}  illustrates scenarios in which  Assumption~\ref{a:app2} does not hold.  Let $b=0.3$, $M=21$, $p=0.5$ and $U(x) = c$ if $x \le 3$ and $U(x)=0$ otherwise.  Figures~\ref{fig:multiple}(b),~\ref{fig:multiple}(d) and ~\ref{fig:multiple}(f) show the three utility functions considered, corresponding to  $c=4$, $12$ and $16$, respectively. If $c=4$, threshold policies with $s \ge 3$ yield optimal reward of zero.  If $c=12$, there are two optimal thresholds ($s=2,3$).  Finally, if $c=16$ the optimal threshold policy is unique ($s=2$).

Proposition~\ref{propomultitude} establishes conditions according to which the number of optimal thresholds is at most two.

\begin{figure}
\center
\includegraphics[scale=0.4]{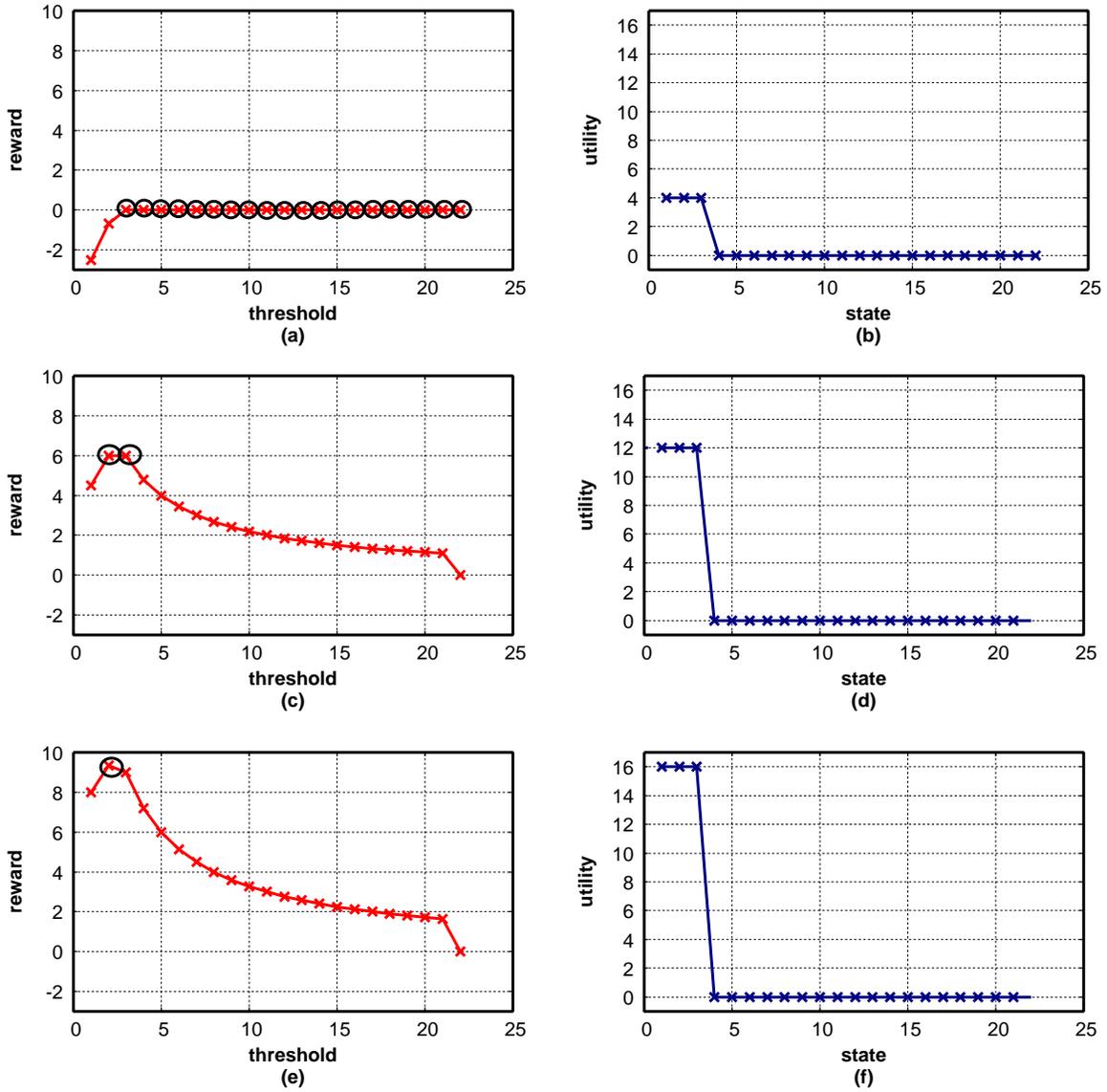}
\caption{Illustration of scenarios in which the optimal threshold is not unique ((a) and (c)) and unique (e).} \label{fig:multiple}
\end{figure}

\begin{proposition}  \label{propomultitude}
Let $R$ be the number of optimal thresholds. 
\begin{equation} \label{eqconditionsmultiple}
 \left\{ \begin{array}{ll} 
R\ge3,  & \textrm{ if } \exists m < {M{-}1} \quad s.t. \quad \sum_{x=1}^{m-1}\utility(x) = G/p+P-B \;\; \textrm{and}\;\;\; \utility(x) = 0 \;\;\;\forall x > m \\
R\leq 2,  & \textrm{ otherwise} 
\end{array} \right.
\end{equation}
If $R \ge 3$ then the policy \emph{always inactive} is optimal.
\label{atmosttwo}
\end{proposition}

\begin{proof}
From proposition~\ref{propos1}, the expected reward  $E[r;s]$ of 
a threshold policy is non-decreasing when $s \leq s^{\star}$, and is non-increasing when $s \ge s^{\star}$. Therefore, the optimal thresholds should be consecutive. Let $s^{\star}-1$, $s^{\star}$ and  $s^{\star}+1$ be three optimal thresholds, 
\begin{equation}
E[r;s^{\star}-1]=E[r;s^{\star}]=E[r;s^{\star}+1] \label{thres3}
\end{equation}
From \eqref{thres3} and \eqref{diffreward}, 
\begin{eqnarray}
&& \left( s^{\star}-1+\frac{1-p}{p} \right) E[r;s^{\star}-1]{-} \left( s^{\star}+\frac{1-p}{p} \right)  E[r;s^{\star}] = - E[r;s^{\star}]  = -p\sum_{i=1}^{M-s^{\star}} \utility(i+s^{\star}-1) (1-p)^{i-1} \label{thres3a}\\
&& \left( s+\frac{1-p}{p} \right) E[r;s^{\star}]{-} \left( s^{\star}+1+\frac{1-p}{p} \right)  E[r;s^{\star}+1] = - E[r;s^{\star}]  = -p\sum_{i=1}^{M-s^{\star}-1} \utility(i+s^{\star}) (1-p)^{i-1} \label{thres3b}
\end{eqnarray}
Subtracting~\eqref{thres3b} from ~\eqref{thres3a},
\begin{equation}
U(s^{\star})=0
\end{equation}
The above equation together with the fact that $U(x)$ is a non-increasing function yields 
\begin{equation} \label{uxequalzero}
U(x)=0, \qquad   x\geq s^{\star}
\end{equation}
~\eqref{uxequalzero} together with ~\eqref{thres3a} and ~\eqref{thres3b} yields
 \begin{equation} \label{uxequalzero2}
 E[r; s^{\star}-1]=E[r; s^{\star}]=E[r; s^{\star}+1]=E[r; M+1] = 0
 \end{equation}
Substituting \eqref{uxequalzero} and \eqref{uxequalzero2}  into \eqref{eq:ers} yields
\begin{eqnarray} \label{uxequalzero3}
\sum_{x=1}^{s^{\star}-1}\utility(x) - \frac{G}{p} - P + B = 0.
\end{eqnarray}
Therefore, if~\eqref{uxequalzero3} and ~\eqref{uxequalzero}  hold then $R \ge 3$. Otherwise, $R \leq 2$.


\end{proof}

\begin{corollary} If the policy \emph{always inactive} is not optimal, the number of optimal thresholds is at most two.
\end{corollary}

\begin{corollary} If the optimal reward is greater than zero the number of optimal thresholds is at most two.
\end{corollary}

\begin{corollary} Let $s^{\star}$ be an optimal threshold.  If $E[r;s^{\star}]\neq E[r;s^{\star}+1]$ and $E[r;s^{\star}]\neq E[r;s^{\star}-1]$ then the optimal threshold is unique.
\end{corollary}

Given the above characterization of the cases in which the optimal threshold policy is not unique, in the rest of this Appendix we  consider the following assumption.

 \begin{assumption} \label{a:app2}
There is at most one optimal threshold policy.
 \end{assumption}

\subsection{Proof of Proposition \ref{propoalin}}

\begin{proof} Let $a$ be the optimal threshold policy.

\subsubsection{Conditions for the optimal policy to be always inactive}

It follows from the optimality conditions~\eqref{eq2} and Assumption~\ref{a:app2} that
 \begin{equation} \label{eqoriginal0b}
 a(x) = 0, 1 \leq x\leq M  \iff H(x,1) \le H(x,0), 1 \leq x \leq M
 \end{equation}
Therefore, from the above equation and ~\eqref{eq2A}-\eqref{eq2B},
 \begin{equation} \label{eqoriginal1b}
 a(x) = 0, 1 \leq x\leq M  \iff V(1)-V(x) \le G/p + P -B, 1 \leq x \leq M
 \end{equation}
    Since $V(x)$ is decreasing (see Proposition~\ref{vdecr_no3g}), ~\eqref{eqoriginal1b} yields
 \begin{equation}
 a(x) = 0, 1 \leq x\leq M  \iff V(1)-V(M) \le G/p + P -B
 \end{equation}
or equivalently,
 \begin{equation}  
 a(x) = 0, 1 \leq x\leq M  \iff V(M) \ge V(1) - G/p - P +B \label{eq:inbelb}
 \end{equation}

From~\eqref{eq2},
 \begin{equation} \label{leftarrowbc}
 a(x) = 0, 1 \leq x\leq M  \iff V(i)=U(i)+V(i+1), 1 \leq i \leq M-1
 \end{equation}
where $\Leftarrow$ in \eqref{leftarrowbc} follows from the assumption that there is a unique optimal policy that satisfies~\eqref{eq2}.  Therefore, 
 \begin{equation} \label{leftarrowb}
 a(x) = 0, 1 \leq x\leq M  \iff V(j)=\sum_{i=j}^M U(i), j=1, \ldots, M
 \end{equation}

Finally, \eqref{eq:inbelb} and \eqref{leftarrowb} yield
 \begin{equation}  
 a(x) = 0, 1 \leq x\leq M  \iff  0 \ge \sum_{i=1}^M U(i) - G/p - P +B 
 \end{equation}

\subsubsection{Conditions for the optimal policy to be always active}

It follows from the optimality conditions~\eqref{eq2} and Assumption~\ref{a:app2} that
 \begin{equation} \label{eqoriginal0}
 a(x) = 1, 1 \leq x\leq M  \iff H(x,1) \ge H(x,0), 1 \leq x \leq M
 \end{equation}
Therefore, from the above equation and ~\eqref{eq2A}-\eqref{eq2B},
 \begin{equation} \label{eqoriginal1}
 a(x) = 1, 1 \leq x\leq M  \iff V(1)-V(x) \ge G/p + P -B, 1 \leq x \leq M
 \end{equation}
    Since $V(x)$ is decreasing (see Proposition~\ref{vdecr_no3g}), it follows from~\eqref{eqoriginal1} that
 \begin{equation}
 a(x) = 1, 1 \leq x\leq M  \iff V(1)-V(2) \ge G/p + P -B
 \end{equation}
or equivalently,
 \begin{equation}  
 a(x) = 1, 1 \leq x\leq M  \iff V(2) < V(1) - G/p - P +B \label{eq:inbel}
 \end{equation}
Letting $x=1$ in~\eqref{eq2} yields
 \begin{equation} \label{leftarrow}
 a(x) = 1, 1 \leq x\leq M  \iff V(1)=-G+U(1)+pV(1) -pP +pB + (1-p)V(2) - E[r;1]
 \end{equation}
where  $\Leftarrow$  in~\eqref{leftarrow} follows since our analysis is restricted to threshold policies.  Note that $a(1)=1$ is implied by the right hand side of~\eqref{leftarrow} together with ~\eqref{eq2}.  If $a(1)=1$ then $a(x)=1, 1\leq x \leq M$.

Finally, \eqref{eq:inbel} and \eqref{leftarrow} yield
 \begin{eqnarray}
&&V(1)=-G+U(1)+pV(1) -pP +pB + (1-p)V(2) - E[r;1] \iff \\
\iff &&V(1)\leq-G+U(1)+pV(1) -pP +pB + (1-p)(V(1)-G/p-P+B) - E[r;1] \label{eqfinalal}
 \end{eqnarray}
The desired result follows from algebraic manipulation of~\eqref{eqfinalal},
 \begin{equation}  
 a(x) = 1, 1 \leq x\leq M  \iff E[r;1] < U(1) - G/p - (P-B)
 \end{equation}
\end{proof}

\begin{remark}
If Assumption \ref{a:app2} does not hold, the proof of Proposition \ref{propoalin} presented above remains valid after replacing all $\iff$ by $\Rightarrow$.
\end{remark}

\subsection{Proof of Proposition~\ref{propos1b}}

\label{app:propos1b}

The proof of Proposition~\ref{propos1b} follows directly from Proposition~\ref{propos1}.

\begin{proposition} \label{propos1}
If $s \leq s^{\star}-1$ then $E[r; s] \leq E[r; s+1]$.  Otherwise, $E[r; s] \geq E[r; s+1]$.
\end{proposition}

\begin{proof}  Next, we show that if $s \leq s^{\star}-1$ then $E[r; s] \leq E[r; s+1]$ (the other case follows similarly).  
Let $s^{\star}$ be the optimal threshold. It follows from Definition~\ref{defopt}  that $E[r; s]- E[r; s^{\star}] \leq 0$ for $1 \leq s \leq M$. Next, we show that
$E[r; s^{\star}-m] \ge E[r; s^{\star}-(m+1)]$ for  $m < s^{\star}$.  The proof is by induction on $m$.

Algebraic manipulation of \eqref{eq:ers} yields
\begin{eqnarray}
&& \left( s-1+\frac{1-p}{p} \right) E[r;s-1]{-} \left( s+\frac{1-p}{p} \right)  E[r;s] \nonumber \\
&=&  \left[ \sum_{i=1}^{s-2} \utility(i)  {+} \sum_{i=0}^{M-s} \utility(i+s-1) (1-p)^{i}  - \frac{\cost}{p} -P+B\right]  - \left[ \sum_{i=1}^{s-1} \utility(i)  {+}  \sum_{i=1}^{M-s} \utility(i+s-1) (1-p)^{i-1}  - \frac{\cost}{p} {-}P+B \right] \nonumber\\
&=& -p\sum_{i=1}^{M-s} \utility(i+s-1) (1-p)^{i-1}. \label{diffreward}
\end{eqnarray}
Equation~\eqref{diffreward} yields
\begin{eqnarray} \label{eq:ers4}
\left({s+\frac{1-p}{p} }\right) (     E[r;s-1] {-} E[r;s] ) = E[r;s]  -p\sum_{i=1}^{M-s} \utility(i+s) (1-p)^{i-1}
\end{eqnarray}

\textbf{Base case: } It follows from Definition~\ref{defopt} that the statement holds for $m=0$.
Note that $E[r; s^{\star}-1] - E[r; s^{\star}] \leq 0$ yields
\begin{eqnarray}   \label{totsum2}
&&\left(s^{\star}-1+\frac{1-p}{p}\right) (E[r; s^{\star}-1] - E[r; s^{\star}])   \\
&\stackrel{(*)}{=}& E[r; s^{\star}] - p\sum_{i=1}^{M-s^{\star}} \utility(i+s^{\star}-1) (1-p)^{i-1} \\
&=& \left(s^{\star}+\frac{1-p}{p}\right) (E[r; s^{\star}-1] - E[r; s^{\star}]) \\
&\stackrel{(**)}{=}&   E[r; s^{\star}-1] - p\sum_{i=1}^{M-s^{\star}} \utility(i+s^{\star}-1) (1-p)^{i-1} \leq 0 \nonumber \\
&\therefore& E[r; s^{\star}-1]  \leq  p\sum_{i=1}^{M-s^{\star}} \utility(i+s^{\star}-1) (1-p)^{i-1} \label{lemmares}
\end{eqnarray}
where $(*)$ follows from ~\eqref{eq:ers} and $(**)$  follows from $(*)$ after summing $E[r; s^{\star}-1] - E[r; s^{\star}]$  to both sides.

\textbf{Induction hypothesis: } Assume that $E[r; s^{*}-m] \ge E[r; s^{*}-(m+1)]$ for  $m < t$.

\textbf{Induction step: } We show that  the proposition holds for $m=t$.
To this goal, we compare $E[r; s^{*}-t]$ and $ E[r; s^{*}-(t+1)]$,
\begin{eqnarray}
&& \left(s^{\star}-t+\frac{1-p}{p}\right) (E[r; s^{\star}-t] - E[r; s^{\star}-t-1])   \nonumber \\
&\stackrel{(*)}{=}&   p\sum_{i=1}^{M-s^{\star}+t} \utility(i+s^{\star}-t-1) (1-p)^{i-1} - E[r; s^{\star}-t] \nonumber \\
&>&  p\sum_{i=1}^{M-s^{\star}} \utility(i+s^{\star}-t-1) (1-p)^{i-1}  -E[r; s^{\star}-t]  \nonumber \\
& >&p\sum_{i=1}^{M-s^{\star}} \utility(i+s^{\star}) (1-p)^{i-1} -E[r; s^{\star}-t]  \nonumber \\
& \stackrel{(**)}{\ge} & p\sum_{i=1}^{M-s^{\star}} \utility(i+s^{\star}) (1-p)^{i-1} -E[r; s^{\star}-1]  \stackrel{(***)}{ \ge} 0 \nonumber
\end{eqnarray}
where $(*)$ follows from ~\eqref{eq:ers4}, $(**)$ follows from the induction hypothesis and ${(}{*}{*}{*}{)}$ follows from ~\eqref{lemmares}.   The proof is completed by noting that $s^{\star}-t+\frac{1-p}{p} \ge 0$ hence $E[r; s^{\star}-t] \ge E[r; s^{\star}-t-1]$.

\end{proof}

\subsection{Proof of Proposition~\ref{increasing}}

\begin{proof}
In what follows we show that the optimal threshold increases with respect to $G$.  The proof that the optimal threshold increases with respect to $P$ and decreases with respect to $B$ is similar. From~\eqref{eqersspecial},

\begin{equation} \label{deriv1}
\frac{d}{dG} E[r;s] = -\frac{1}{ps+1-p}
\end{equation}
Let $s_1 > s_0$.  Then,
\begin{eqnarray} 
\lim_{\Delta G \rightarrow 0} \frac{ E[r;s_1,G+\Delta G] - E[r;s_1,G]}{\Delta G}    >  \lim_{\Delta G \rightarrow 0} \frac{E[r;s_0,G+\Delta G] - E[r;s_0,G]}{\Delta G}
\end{eqnarray}

From~\eqref{eqersspecial}, it also follows that
\begin{equation}
E[r;s,G+\Delta G] - E[r;s,G] = -\frac{1}{s+(1-p)/p} \frac{\Delta G}{p}
\end{equation}

Therefore,
\begin{equation} \label{eqless}
E[r;s_1,G] - E[r;s_1,G+\Delta G]  < E[r;s_0,G]- E[r;s_0,G+\Delta G]
\end{equation}

Assume, for the sake of contradiction, that $s_1$ and $s_0$ are optimal thresholds when the energy cost is $G$ and $G+\Delta$, respectively,
\begin{eqnarray}
E[r;s_1,G] &\ge& E[r;s,G], \qquad \qquad s \neq s_1 \\
E[r;s_0,G+\Delta G] &\le& E[r;s,G+\Delta G], \quad s \neq s_0
\end{eqnarray}
In particular,
\begin{eqnarray}
E[r;s_1,G] &\ge& E[r;s_0,G]\\
E[r;s_0,G+\Delta G] &\le& E[r;s_1,G+\Delta G]
\end{eqnarray}
  Then, $E[r;s_1,G] - E[r;s_1,G+\Delta G]  \ge E[r;s_0,G]- E[r;s_0,G+\Delta G]$, which contradicts the ~\eqref{eqless}.

\end{proof}

\begin{remark}
In~\eqref{eqless}, $<$ should be replaced by $\le$ if Assumption~\ref{a:app2} does not hold. In this case, the optimal threshold in non decreasing with respect to $G$.
\end{remark}

\subsection{Proof of Proposition~\ref{propo:step} (step utility)} 
\label{app:step}

\begin{proof}
We assume that the optimal policy consists of remaining active in at least one state (otherwise, $s^{\star}=M+1$ and $E[r;s]=0$).  We consider two cases,

\underline{$i)$ $k > s$}:
\begin{equation} \label{toderive1}
E[r;s]=  \frac{1}{s+\frac{1-p}{p} } \left[ v k  - \frac{\cost}{p} -P +B\right]
\end{equation}

\underline{$ii)$  $k \leq s$}:  
\begin{eqnarray}
E[r;s]&=& \frac{1}{s+\frac{1-p}{p} } \left[ vs  + v \sum_{i=1}^{k-s}  (1-p)^{i}  - \frac{\cost}{p} -P \right]  =  \label{toderive2} 
 \frac{1}{s+\frac{1-p}{p} } \left[ vs  + \frac{v(1-p)}{p} -\frac{v(1-p)^{k-s+1}}{p} - \frac{\cost}{p} -P +B \right] \nonumber
\end{eqnarray}
We now show that in the above two cases the optimal threshold can be efficiently computed by comparing the values of $E[r;s]$ at five points.   To this goal, we first assume that the optimal threshold, $s^{\star}$,  can take real values (the assumption will be removed in the next paragraph). We refer to $s^{\star}$ as an interior maximum if $1 < s^{\star} < M$, and as a boundary maximum if $s^{\star}=1$ or $s^{\star}=M$.  Then, a necessary condition for $s^{\star}$ to be an interior maximum of $E[r;s]$ consists of $s^{\star}$  being a root of $\frac{d}{ds} E[r;s]=0$.    If $k > s$ then  $\frac{d}{ds} E[r;s]=0$  has no roots in the interval $[1,M]$, since~\eqref{toderive1} is monotonic with respect to $s$. If $k \leq s$,  let $\varphi$ be  the root of $\frac{d}{ds} E[r;s]=0$. Then, 
\begin{eqnarray}
\varphi &=& - \left( p+{W} \left(  \left( G+pP-pB \right) {\exp\left({-{\frac {
\ln  \left( 1-p \right)(1+ kp)+p }{p}}}\right)}{v}^{-1}
 \right) p+\ln  \left( 1-p \right) (1- p )\right)
\Big/ \left( \ln  \left( 1-p \right)  \right) {p} \label{varphi}
\end{eqnarray}
In~\eqref{varphi}, $W(x)$ denotes the  Lambert function, i.e., $W(x)=w$ if $w e^w  = x$.

Accounting for the fact that $s^{\star}$ is an integer, and that it might be either an interior or boundary maximum, yields,
\begin{equation}
s^{\star} {=}
\left\{ \begin{array}{ll} 1, & \textrm{if } k > s^{\star} \textrm{ and } vk-G/p-P+B > 0 \\
k-1, & \textrm{if }  k > s^{\star} \textrm{ and } vk-G/p-P+B \le 0 \\
\psi = \min(\lceil {\varphi} \rceil,M), & \textrm{if }  k \leq s^{\star} \textrm{ and } E[r;\psi] {\ge} E[r;\psi+ 1]  \\
\min(\lfloor  {\varphi} \rfloor,M), & \textrm{otherwise} \\
 \end{array} \right.
\end{equation}
\end{proof}
Therefore, $s^{\star}$ can be computed comparing the value of $E[r;s]$ at $s=1, k-1, \lfloor \varphi \rfloor, \lceil \varphi \rceil, M$.

\subsection{Linear utility}\label{app:linear}

When $U(x)=M-x$, the expected reward is
\begin{eqnarray}
E[r;s]&=& \pi_1 \Big[ (s-1)M - s(s-1)/2  + \frac{(1-p)^{M-s+1}+pM-ps-1+p}{p^2}  - \frac{\cost}{p} -P +B \Big] \label{ers1}  \nonumber
\end{eqnarray}

\subsection{Expected age}
\label{app:expectedage}

If $p>0$, the expected age, $A$, is $A=\sum_{i=1}^M i \pi_i$,

\begin{eqnarray}
A &=& \pi_1 \Big[ \sum_{i=1}^{s}{i} + \sum_{i=s+1}^{M-1}{i (1-p)^{i-s}} \nonumber + \frac{(1-p)^{M-s}}{p} {M} \Big] \\
&=&\sum_{i=1}^M i \pi_i= M \frac{1-p}{p} (1-p)^{M-1-s} \frac{1}{s+(1-p)/p} +  \sum_{i=1}^{s}  i \frac{1}{s+(1-p)/p} +  \sum_{i=s+1}^{M-1} (1-p)^{i-s}i\frac{1}{s+(1-p)/p} \nonumber \\
&=&\,{\frac {{p}^{2}{s}^{2}-{p}^{2}s-2\, \left( 1-p \right) ^{M-s} (1-p)+2\,sp+2-2\,p}{2 p \left( sp+1-p \right) }} \label{expectedage}
\end{eqnarray}

\subsection{Proposition~\ref{ageincrs}}

\begin{proposition} The expected age is an increasing function of $s$. \label{ageincrs}
\end{proposition}

\begin{proof}
Let $A(s)$ be the average age of a threshold policy whose threshold is $s$.
Consider two threshold policies, with corresponding thresholds $s$ and $s+1$.  Then,
\begin{eqnarray}
A(s) &=& \sum_{i=1}^{s}\frac{i}{s+\frac{1-p}{p}} + \sum_{i=s+1}^{M-1}\frac{i (1-p)^{i-s}}{s+\frac{1-p}{p}}  + \frac{(1-p)^{M-s}}{p} \frac{M}{s+\frac{1-p}{p}} \label{eq:as} \\
A(s+1) &=& \sum_{i=1}^{s+1}\frac{i}{s+1+\frac{1-p}{p}} + \sum_{i=s+2}^{M-1}\frac{i (1-p)^{i-s-1}}{s+1+\frac{1-p}{p}}  + \frac{(1-p)^{M-s-1}}{p} \frac{M}{s+1+\frac{1-p}{p}} \label{eq:as1} 
\end{eqnarray}
Algebraic manipulation of~\eqref{eq:as}-\eqref{eq:as1} yields
\begin{eqnarray}
&& A(s)  \frac{s+\frac{1-p}{p}}{s+1+\frac{1-p}{p}} = \sum_{i=1}^{s}\frac{i}{s+1+\frac{1-p}{p}} +  \sum_{i=s+1}^{M-1}\frac{i (1-p)^{i-s}}{s+1+\frac{1-p}{p}} + \frac{(1-p)^{M-s}}{p} \frac{M}{s+1+\frac{1-p}{p}} \nonumber \\ \label{subtr1} \\
&&  A(s+1)  (1-p) = \sum_{i=1}^{s+1}\frac{i(1-p)}{s+1+\frac{1-p}{p}} + \sum_{i=s+2}^{M-1}\frac{i (1-p)^{i-s}}{s+1+\frac{1-p}{p}} + \frac{(1-p)^{M-s}}{p} \frac{M}{s+1+\frac{1-p}{p}} \nonumber \\ \label{subtr2}
\end{eqnarray}
Subtracting the left-hand side of ~\eqref{subtr2}  from the left-hand side of ~\eqref{subtr1}, and denoting the difference by $\Delta_1$,  yields
\begin{eqnarray}
\Delta_1&=& A(s)  \frac{s+\frac{1-p}{p}}{s+1+\frac{1-p}{p}}  - A(s+1)  (1-p) \nonumber\\
&=& A(s)-A(s+1)-A(s)\frac{1}{s+1+\frac{1-p}{p}}  + A(s+1) p \nonumber\\
& =& (A(s) -A(s+1))(1-p) + A(s) \frac{ps}{s+1+\frac{1-p}{p}}.
\end{eqnarray}
Subtracting the right-hand side of ~\eqref{subtr2}  from the right-hand side of ~\eqref{subtr1}, and denoting the difference by $\Delta_2$, yields
\begin{eqnarray}
\Delta_2 = \sum_{i=1}^{s}\frac{pi}{s+1+\frac{1-p}{p}}.
\end{eqnarray}
Letting $\Delta_1=\Delta_2$,
\begin{equation}
(A(s) -A(s+1))(1-p) {=} \frac{ps}{s+1+\frac{1-p}{p}}  \left( \frac{1}{s} \sum_{i=1}^{s}i -  A(s) \right)
\label{eqkeydecre}
\end{equation}

  Next, we make the dependence of the age on $p$ explicit, and denote it by $A(s,p)$.  
  
  The age is minimized when $p=1$,
\begin{equation} 
A(s,1) < A(s,p), \qquad 0 < p < 1
\end{equation}

  Letting $p=1$ in \eqref{expectedage},
\begin{eqnarray} \label{p1final}
A(s,1) = \frac{1}{s}\sum_{i=1}^{s}i.
\end{eqnarray}
~\eqref{p1final}  implies that the right hand side of ~\eqref{eqkeydecre} is negative.  
Therefore, it follows from \eqref{eqkeydecre}-\eqref{p1final} that $A(s) < A(s+1)$ and, for $s_1 < s_2$, $A(s_1) < A(s_2)$.
\end{proof}

\subsection{Proof of Proposition~\ref{optimalpubprob}}

\begin{proof}
Let $f$ be a function that maps a bonus level into the corresponding optimal threshold selected by users that solve the {\ensuremath{\mbox{\sc  User Problem}}}, 

\begin{eqnarray}
f: B & \rightarrow& s^{\star} \\
\mathbb{R} & \rightarrow& \mathbb{N}
\end{eqnarray}

Let $g$ be a function that maps a threshold $s$ into the \emph{minimum} bonus level $B$ such that  $s$ is the optimal threshold under $B$,

\begin{eqnarray}
g: s & \rightarrow& B \\
\mathbb{N} & \rightarrow& \mathbb{R}
\end{eqnarray}

Let $h$ be a function that maps a threshold $s$ into the pair of  bonus levels $(\underline{B}(s), \overline{B}(s))$ such that  $s$ is the optimal threshold under any bonus in the range $[\underline{B}(s), \overline{B}(s)]$,

\begin{eqnarray}
h: s & \rightarrow& (\underline{B}(s), \overline{B}(s)) \\
\mathbb{N} & \rightarrow& \mathbb{R}^2
\end{eqnarray}

Note that $f$ admits no inverse (changing $B$ might not alter $s^{\star}$).  The function $g$, in contrast, is injective, since given  $B$ the optimal threshold $s^{\star}(B)$ is unique.

The publisher must choose $B$ so as to minimize $A$.  We consider two cases, varying according to the value of $s(B)$ evaluated at  $B=P$,

\underline{case i)}  

$$s(P) < \frac{Np}{T}-\frac{1-p}{p}$$

  In this case, inequality~\eqref{ineq2} is not satisfied when $B = P$. Since $B$ cannot surpass $P$, the problem admits no solution.

\bigskip
\underline{case ii)} 

$$s(P) \ge \frac{Np}{T}-\frac{1-p}{p}$$ 
In this case, the problem admits a solution, to be derived next.

According to Proposition~\ref{ageincrs}, the average age $A$ is an increasing function of $s$. According to Proposition \ref{increasing}, the optimal threshold $s^{\star}$ is
a non-increasing function of $B$.   Therefore, the objective function $A$ of the {\ensuremath{\mbox{\sc  Publisher Problem}}}  is a non-increasing function of $B$.

The fact that $A$ is a non-increasing function of $B$, together with the fact that the l.h.s. of inequality~\eqref{ineq2} is increasing in $B$,  imply that the optimal bonus consists of setting $B$  as large as possible, without violating~\eqref{ineq2}.

Inequality~\eqref{ineq2} binds when
\begin{eqnarray}
s =  \frac{Np}{T}-\frac{1-p}{p}
\end{eqnarray}
Therefore, the optimal threshold, which must be integer, is 
\begin{eqnarray}
s = \min \left(M+1, \left\lceil \frac{Np}{T}-\frac{1-p}{p} \right\rceil  \right)
\end{eqnarray}
The solution of the {\ensuremath{\mbox{\sc  Publisher Problem}}} consists of finding $B$ such that $s^{\star}=s$.  

The bonus $B$ must be chosen in the range $[\underline{B}(s^{\star}),\overline{B}(s^{\star})]$.  Since $s^{\star}$ is a non-increasing function of $B$ (Proposition \ref{increasing}), the value of $B$ which yields $s^{\star} = \min \left(M+1, \left\lceil \frac{Np}{T}-\frac{1-p}{p} \right\rceil  \right)$  can be found using a binary search algorithm.

\end{proof}

\section{Additional results when $P_{3G} \rightarrow \infty$}

\subsection{Matrix form and linear program}

\begin{equation}
A_1=\left(
\begin{array}{ccccccccccccc}
1  &  -1 &  &        &   & \\
0  &  1 & -1&        &   & \\
0  &  0 & 1& -1      &   & \\
   &    &  & \ddots &   & \\
     &    &  &        & 1 & -1 & \\
   &    &  &   & & \ddots &   & \\
  &    &  &   & &      & 1 & -1 & \\
  &   &   &     & &   &    0  &    0 \\
 \end{array}
\right)
\end{equation}

\begin{equation}
A_2=\left(
\begin{array}{ccccccccccc}
-p +1 & -(1-p)  &  &        &   & &  &  \\
-p &  1 &  -(1-p) &        &   & \\
-p  &  0 & 1&   -(1-p)    &   & \\
   &    &  & \ddots &   & \\
     &    &  &        & 1 & -(1-p) & \\
   &    &  &   & & \ddots &   & \\
-p  &    &  &   & &      & 1 & -(1-p) & \\
-p  &   &   &     & &   &    0  &    -p \\
 \end{array}
\right)
\end{equation}

Let $V$ be the solution of the optimality equations.

\begin{eqnarray}
A_1 V &\ge& U \\
A_2 V &\ge& U - G + pP - pB 
\end{eqnarray}
Letting $\widetilde{A_1}$ and $\widetilde{A_2}$  denote matrices $A_1$ and $A_2$ with their last line and column removed, yields the following LP 
\begin{eqnarray}
\min & |\widetilde{V}| \\
\widetilde{V} &\ge& \widetilde{A_1}^{-1} U \\
\widetilde{V} &\ge& A_2^{-1} ( U - G + pP - pB )
\end{eqnarray}
where $\widetilde{V}$ is the vector $V$ with its last element removed.

\section{Derivation of Main Results Accounting for 3G and WiFi}
\label{sec:w3g}

We now analyze the case in which 3G is available.  The optimality conditions are, for $1 \leq x \leq M$,
\begin{eqnarray}   \boxed{ \boxed{ 
V(x)+g_u = \max \Big( U(x) {+}K(x),U(x) {-} G {+} pV(1) {+} (1-p) K(x) {-}p P {+}pB, U(x) {-}G {+} V(1){-}p P  {-}(1-p)P_{3G} {+} B \Big) }} \label{with3gop}
\end{eqnarray}
where $K(x)=V(\min(M,x+1))$.

\bigskip

Let $F(a,x)$ be the relative reward at state $x$ when action $a$ is chosen, plus $g_u$,
\begin{eqnarray}
F(0,x)&=& U(x) +V(\min(M,x+1))\label{fct1} \\
F(1,x)&=& U(x) - G + pV(1) + (1-p) V(\min(M,x+1)) -p P +pB \label{fct2}\\
F(2,x)&=& U(x) -G + V(1)-p P  -(1-p)P_{3G} + B.\label{fct3}
\end{eqnarray}
The optimality  conditions are, for $x=1,\ldots, M$,
\begin{eqnarray}
V(x)= \max\left( F(0,x)-g_u ,F(1,x)-g_u,F(2,x)-g_u \right). \label{fct4}
\end{eqnarray}
Let $\mathcal{F}(a,x)$, $a=0,1,2$, and $x=1,\ldots, M$, be boolean variables such that
\begin{eqnarray}
\mathcal{F}(0,x) = \top &\Leftrightarrow& F(0,x) \ge \max_i(F(i,x))\\
\mathcal{F}(1,x) = \top &\Leftrightarrow& F(1,x) \ge \max_i(F(i,x))\\
\mathcal{F}(2,x) = \top &\Leftrightarrow& F(2,x) \ge \max_i(F(i,x))\\
\end{eqnarray}

From~\eqref{fct1}-\eqref{fct3}, a policy that selects $a(x)$, $x=1,\ldots, M$, as follows, is optimal 

\begin{equation} 
\left\{ \begin{array}{llll} 0,  & \textrm{if }  V(1) {-} V(\min(M,x+1)) \leq {G}/{p} + P - B   & \textrm{ and } V(1) - V(\min(M,x+1)) \leq G + pP + (1-p)P_{3G} {-} B  & (i) \\ 
1,  & \textrm{if } V(1) {-} V(\min(M,x+1)) \ge {G}/{p} + P - B & \textrm{ and } V(1) - V(\min(M,x+1)) \leq P_{3G} - B & (ii) \\
2, & \textrm{if } V(1) {-} V(\min(M,x+1)) \ge P_{3G} - B  & \textrm{ and } V(1) - V(\min(M,x+1)) \ge G + pP + ( 1-p)P_{3G} {-} B \label{opt_cond3} & (iii)  \\
 \end{array} \right.
\end{equation}
where ties are broken arbitrarily.

\subsection{Proposition~\ref{prop:decr_w3g}}

\label{app:vdecrw3g}
\begin{proposition} \label{prop:decr_w3g}
$V(x)$ is a non-increasing function.
\end{proposition}

\begin{proof}
According to ~\eqref{opt_cond3}-(ii),  if $P_{3G}< {G}/{p}+P$, there exists an optimal  policy in which action 1 is not selected.  Therefore, in what follows we separately consider two scenarios,
$P_{3G}\leq {G}/{p}+P$ and $P_{3G}> {G}/{p}+P$.

\bigskip
\begin{center}
\begin{equation}
\boxed{ \boxed{
\textrm{{scenario 1)  } } P_{3G} \leq \frac{G}{p}+P} }
\end{equation}
\end{center}
\bigskip

In this scenario, there exists an optimal policy wherein $a(x) \neq 1$, $x=1, \ldots,M$. We show by induction that 
\begin{itemize}
\item $V(M-i-1) - V(M-i) \ge 0$ for $i=0, \ldots, M-1$;
\item $\mathcal{F}(M-i,0) = \top \Rightarrow \mathcal{F}(M-i-1,0)=\top$ for $i=0, \ldots, M-1$.
\end{itemize}

\eqref{opt_cond3}(i) and \eqref{opt_cond3}(iii) yield the following remark,  used in the analysis of the  base cases of the  inductive arguments that follow.
\begin{remark} $\mathcal{F}(M,0)=\top   \iff \mathcal{F}(M-1,0) =\top$.  \label{remark1w3g} $\square$
\end{remark}

We consider two cases, $\mathcal{F}(M,0)=\top$ and $\mathcal{F}(M,2)=\top$.

\begin{itemize}
\item \underline{$\mathcal{F}(M,0)=\top$:} 

\textbf{Base case: } We first show that $V(M-1) \ge V(M) \label{ineq1}$. 
 It follows from  ~\eqref{fct4} and remark~\ref{remark1w3g} that
 \begin{eqnarray}
 V(M-1)&=&\utility(M-1)+V(M)-g_{u}  \label{base1H} \\
 V(M) &=& \utility(M)+V(M)-g_{u}\nonumber\label{base2HH}\\
 &\ge& \utility(M) + V(1) -pP -G -(1-p) P_ {3G} + B  -g_u. \label{base2H}
 \end{eqnarray}
Hence, \eqref{base1H}-\eqref{base2H} together with the fact that $U(M)$ is non-increasing and $U(M)=0$ yield 
$$V(M-1) \ge V(M)\geq V(1) -pP -G -(1-p) P_ {3G} + B - g_u.$$\\
\textbf{Induction hypothesis: } Assume that $V(M-m-1) - V(M-m) \ge 0$ for $m < t$ and that $\mathcal{F}(M-m,0) = \top$,  for $m \leq t$.\\
\textbf{Induction step: } Next, we show that $V(M-t-1) - V(M-t) \ge 0$ and that $\mathcal{F}(M-t-1,0)=\top$.

It follows from the induction hypothesis that 
\begin{equation}
V(1)-pP -G +(1-p)P_ {3G} + B\leq V(M-t+1) \leq V(M-t)
\end{equation}
 and $\mathcal{F}(M-t,0)=\top$.  Therefore,
 \begin{equation} \label{ineqright}
V(1) - V(M-t) \leq  V(1)-V(M-t+1)  \leq pP +G +(1-p)P_{3G} - B
 \end{equation}
 The rightmost inequality in~\eqref{ineqright}  and optimality conditions~\eqref{fct4} imply that $\mathcal{F}(M-t-1,0)=\top$, which yield
\begin{eqnarray}
V(M-t-1) &=& \utility(M-t-1) + V(M-t) - g_{u} \label{eq1aa} \\
V(M-t) &=& \utility(M-t) + V(M-t+1) - g_{u} \label{eq1ba}
\end{eqnarray}
Therefore, $V(M-t+1) \leq V(M-t)$ (induction hypothesis) and $\utility(M-t) \geq \utility(M-t-1)$ together with~\eqref{eq1aa}-\eqref{eq1ba} yield $V(M-t) \leq V(M-t-1)$.\\


\item \underline{$\mathcal{F}(M,2)=\top$:} 

\textbf{Base case: } We first show that $V(M-1) \ge V(M) \label{ineq1H}$.  It follows from the optimality conditions~\eqref{fct4}  and remark~\ref{remark1w3g} that
 \begin{eqnarray}
 V(M-1)&=& U(M-1)+ V(1)- G -p P-(1-p) P_{3G}-g_{u} + B \label{base1bH} \\
 V(M) &=&U(M)+V(1)- G -p P-(1-p) P_{3G}-g_{u} + B. \label{base2bH}
 \end{eqnarray}
 Hence, \eqref{base1bH}-\eqref{base2bH} and $U(M)<U(M-1)$ yield $V(M-1) \ge V(M)$.\\
\textbf{Induction hypothesis: } Assume that $V(M-m-1) - V(M-m) \ge 0$ for $m < t$ and that $\mathcal{F}(M-m+1,0) =\top \Rightarrow \mathcal{F}(M-m,0) =\top$, for $m \leq t$.\\
\textbf{Induction step: } Next, we show that $V(M-t-1) - V(M-t) \ge 0$ and $\mathcal{F}(M-t,0) =\top  \Rightarrow \mathcal{F}(M-t-1,0) =\top$.  We consider three cases,
\begin{itemize}
\item  \underline{$\mathcal{F}(M-t,0)=\top$}
 This case is similar to the corresponding one for $\mathcal{F}(M,0)=\top$.

\item \underline{$\mathcal{F}(M-t,2)=\top$ and $\mathcal{F}(M-t-1,2)=\top$}   Next, we show that $V(M-t-1) - V(M-t) \ge 0$. 
\begin{eqnarray}
V(M-t) &=& \utility(M-t)+V(1) - G -pP -(1-p) P_{3G}-g_{u}  + B \label{eq3a} \\
V(M-t-1) &=& \utility(M-t-1)+V(1) - G -p P -(1-p) P_{3G}-g_{u} + B.  \label{eq3b}
\end{eqnarray}
Since $\utility(M-t) \leq \utility(M-t-1)$, it follows from~\eqref{eq3a}-~\eqref{eq3b} that $V(M-t) \leq V(M-t-1)$. 

\item \underline{$\mathcal{F}(M-t,2)=\top$ and $\mathcal{F}(M-t-1,0)=\top$}   Next, we show that $V(M-t-1) - V(M-t) \ge 0$. 
\begin{eqnarray}
V(M-t) &=& \utility(M-t)+V(1) - G -pP -(1-p) P_{3G}-g_{u}  + B  \\
V(M-t-1) &=& \utility(M-t-1)+V(M-t) \label{vmt1}
\end{eqnarray}
Since $\utility(M-t-1) \ge 0$, it follows from~\eqref{vmt1} that $V(M-t) \le V(M-t-1)$. 

\end{itemize}
\end{itemize}

\bigskip
\begin{center}
\begin{equation}
\boxed{ \boxed{
\textrm{{scenario 2)  } } P_{3G} > \frac{G}{p}+P} }
\end{equation}
\end{center}
\bigskip

We show by induction that, for $i=0,..,M-1$,
\begin{itemize}
\item $V(M-i-1) \ge V(M-i)$, 
\item $\mathcal{F}(M-i,0)= \top\Rightarrow \mathcal{F}(M-i-1,0)=\top$
\item $\mathcal{F}(M-i,1)= \top \Rightarrow (\mathcal{F}(M-i-1,1)=\top \mbox{ or } \mathcal{F}(M-i-1,0)=\top)$
\end{itemize}
We consider three cases, $\mathcal{F}(M,0)=\top$, $\mathcal{F}(M,1)=\top$ and $\mathcal{F}(M,2)=\top$.

\eqref{opt_cond3}(i), \eqref{opt_cond3}(ii) and \eqref{opt_cond3}(iii), together, yield the following remarks,  used in the analysis of the  base cases of the  inductive arguments that follow.
\begin{remark} $\mathcal{F}(M,1)=\top   \iff \mathcal{F}(M-1,1) =\top$.  \label{remark1w3gb} $\square$
\end{remark}

\begin{remark} $\mathcal{F}(M,2)=\top   \iff \mathcal{F}(M-1,2) =\top$.  \label{remark1w3gc} $\square$
\end{remark}

\begin{itemize}
\item \underline{$\mathcal{F}(M,0)=\top$:} The proof is similar to that of scenario 1.
\item \underline{$\mathcal{F}(M,1)=\top$:} 

\textbf{Base case: }
It follows from remark~\ref{remark1w3gb} and  ~\eqref{fct1}-\eqref{fct4} that
 \begin{eqnarray}
 V(M-1)&=& U(M-1)+p V(1)-p P + pB -(1-p)V(M)-g_{u} \label{base3bH} \\
 V(M) &=&U(M)+p V(1)-p P + pB -(1-p)V(M)-g_{u}.\label{base4bH}
 \end{eqnarray}
 Hence, inequalities in \eqref{base3bH}-\eqref{base4bH} yield 
 $$V(M-1) \geq V(M)$$

\textbf{Induction hypothesis: } Assume that, for $m=0,..,t-1$,
\begin{itemize}
\item[a)] $V(M-m-1)\ge V(M-m)$
\item[b)] $\mathcal{F}(M-m,0)= \top\Rightarrow \mathcal{F}(M-m-1,0)=\top$
\item[c)] $\mathcal{F}(M-m,1)= \top \Rightarrow (\mathcal{F}(M-m-1,1)=\top \mbox{ or } \mathcal{F}(M-m-1,0)=\top)$
\end{itemize}

\textbf{Induction step: } Next, we show that 
\begin{itemize}
\item[a)] $V(M-t-1) \ge V(M-t)$
\item[b)] $\mathcal{F}(M-t,0)= \top\Rightarrow \mathcal{F}(M-t-1,0)=\top$
\item[c)] $\mathcal{F}(M-t,1)= \top \Rightarrow (\mathcal{F}(M-t-1,1)=\top \mbox{ or } \mathcal{F}(M-t-1,0)=\top)$
\end{itemize}

 It follows from the induction hypothesis that $V(M-t+1) \leq V(M-t)$.   We consider three cases,
\begin{itemize}
\item \underline{$\mathcal{F}(M-t,0)=\top$} The proof is similar to the corresponding one for $H(M,1) \le H(M,0)$ in the proof of Proposition~\ref{vdecr_no3g}.
\item \underline{$\mathcal{F}(M-t,1)=\top$} The proof is similar to the corresponding one for  $H(M,1) \ge H(M,0)$ in the proof of Proposition~\ref{vdecr_no3g}.
\item \underline{$\mathcal{F}(M-t,2)=\top$}  Due to the induction hypothesis (cases b) and c)), one of the two cases above holds.  Hence, it is not necessary to consider the case $\mathcal{F}(M-t,2)=\top$.
\end{itemize}

\item \underline{$\mathcal{F}(M,2)=\top$:} 

\textbf{Base case: } It follows from remark~\ref{remark1w3gc} that $V(M-1) \leq V(M)$.   

\textbf{Induction hypothesis: } Assume that, for $m=0,..,t-1$,
\begin{itemize}
\item $V(M-m-1)\ge V(M-m)$
\item $\mathcal{F}(M-m,0)=\top\Rightarrow \mathcal{F}(M-m-1,0)=\top$
\item $\mathcal{F}(M-m,1)=\top\Rightarrow (\mathcal{F}(M-m-1,1)=\top \mbox{ or } \mathcal{F}(M-m-1,0)=\top)$
\end{itemize}

\textbf{Induction step: } Next, we show that 
\begin{itemize}
\item $V(M-t-1) \ge V(M-t)$
\item $\mathcal{F}(M-t,0)=\top\Rightarrow \mathcal{F}(M-t-1,0)=\top$
\item $\mathcal{F}(M-t,1)=\top\Rightarrow (\mathcal{F}(M-t-1,1)=\top \mbox{ or } \mathcal{F}(M-t-1,0)=\top)$
\end{itemize}

We consider the following cases.
\begin{itemize}
\item \underline{$\mathcal{F}(M-t,0)=\top$} The proof is similar to that of scenario 1.

\item \underline{$\mathcal{F}(M-t,1)=\top$} First, we show that
\begin{equation}
(\mathcal{F}(M-t-1,1) =\top )  \textrm{ or }  (\mathcal{F}(M-t-1,0) =\top  )
\end{equation}
 For the sake of contradiction, assume 
 \begin{equation} \label{assumptionsakecontr}
(\mathcal{F}(M-t-1,2) =\top )  \textrm{ and  } (  (\mathcal{F}(M-t-1,0) \neq \top  ) \textrm{ and } (\mathcal{F}(M-t-1,1) \neq \top ) ) 
\end{equation}
Then,
\begin{eqnarray}
F(M-t-1,2) &\ge& F(M-t-1,1) \label{fmt12a} \\
F(M-t,1) &\ge& F(M-t,2) \label{fmt12b}
\end{eqnarray}
Replacing ~\eqref{fct1}-\eqref{fct3} into ~\eqref{fmt12a}-\eqref{fmt12b},
\begin{eqnarray} 
-G+V(1)-pP-(1-p)P_{3G}+B   &\ge& -G+pV(1)+(1-p)V(M-t)-pP+pB 
\label{useful1} \\
-G+pV(1)+(1-p) V(M{-}t+1) - pP{+} p B &\ge&  - G+ V(1) - pP - (1-p)P_{3G} + B  \label{useful2} 
\end{eqnarray}

From \eqref{useful1} and \eqref{useful2}  we obtain, respectively,
\begin{eqnarray}
V(1) - V(M-t)      -P_{3G} +B  &\ge& 0 \label{inducthypotapplhere} \\
-V(1) + V(M{-}t+1)  +P_{3G} -B  &\ge& 0 \label{v1vmt1p3ga}
\end{eqnarray}
Applying the induction hypothesis to~\eqref{inducthypotapplhere},
\begin{eqnarray} \label{v1vmt1p3gb}
V(1) - V(M-t+1)  -P_{3G} +B  &\ge& 0  
\end{eqnarray}
Therefore, from~\eqref{v1vmt1p3ga} and \eqref{v1vmt1p3gb},
\begin{eqnarray}
 V(M-t+1)  +P_{3G} -B  &=&  V(1) \label{nonuniqueness}
\end{eqnarray}	
~\eqref{nonuniqueness} together with ~\eqref{fct1}-\eqref{fct3} imply that both $\mathcal{F}(M-t+1,2)=\top$ and $\mathcal{F}(M-t+1,1)=\top$  are optimal, which contradicts~\eqref{assumptionsakecontr}.

To show that  $V(M-t-1) \ge V(M-t)$, an argument similar to the corresponding one for $H(M,1)\ge H(M,0)$ in the proof of Proposition~\ref{vdecr_no3g} can be applied.

\item \underline{$\mathcal{F}(M-t,2)=\top$ } We consider three subcases.
\begin{itemize}
\item \underline{$\mathcal{F}(M-t,2)=\top$ and $\mathcal{F}(M-t-1,2)=\top$} 
 The proof is similar to the corresponding case in scenario 1.
\item \underline{$\mathcal{F}(M-t,2)=\top$ and $\mathcal{F}(M-t-1,0)=\top$}  The proof is similar to the corresponding case in scenario 1.
\item \underline{$\mathcal{F}(M-t,2)=\top$ and $\mathcal{F}(M-t-1,1)=\top$}  In this case, it follows from $F(1,M-t-1) \ge F(2,M-t-1)$  that
\begin{eqnarray}
-G+pV(1)+(1-p) V(M{-}t) - pP{+} p B &\ge&  - G+ V(1) - pP - (1-p)P_{3G} + B \label{last}
\end{eqnarray}
From~\eqref{last},
\begin{equation}
pV(1)+(1-p) V(M-t) + p B \ge   V(1) - (1-p)P_{3G} + B
\end{equation}
Therefore, 
\begin{equation}
U(M-t-1)+pV(1)+(1-p) V(M-t) + p B \ge U(M-t)+V(1)-(1-p)P_{3G}+B  \label{vmt1lessvmt}
\end{equation}
~\eqref{vmt1lessvmt} together with ~\eqref{fct1}-\eqref{fct3} imply that $V(M-t-1) \ge V(M-t)$.
\end{itemize}

\end{itemize}

\end{itemize}

\end{proof}

\subsection{Proposition~\ref{temp1}}

\begin{proof}
The proof follows from Proposition~\ref{prop:decr_w3g} and optimality  conditions \eqref{opt_cond3}.
\end{proof}


\subsection{Proposition \ref{allinact_w3g} (conditions for the optimal policy to consist of using only one action)}

\begin{proposition}  \label{allinact_w3g}
Let $a(x)=\alpha$, $x=1, \ldots, M$ be an optimal policy. 

If $P_{3G} \leq\frac{G}{p}+P$ and\\
\begin{itemize}
\item $\alpha=2$,
\begin{equation}
U(1)-U(2) \ge G+(1-p) P_{3G} + pP  - B. \label{cond1_w3g}
\end{equation}
\item $\alpha=0$,
\begin{equation}
\sum_{i=1}^{M-1} U(M-j) \le G+ (1-p) P_{3G} +p P - B. \label{cond2_w3g}
\end{equation}
\end{itemize}
If $P_{3G} > \frac{G}{p}+P$ and \\
\begin{itemize}
 \item $\alpha=0$,
\begin{equation}
 \sum_{j=1}^{M-1} U(M-j) \le \frac{G}{p} +P - B. \label{cond3_w3g}
\end{equation}
\item $\alpha=1$,
\begin{eqnarray}
U(1) - p\sum_{i=1}^{M-1}\utility(i)(1-p)^{i-1} &\geq& (1-p)(G/p + P - B) \label{cond4_w3g}\\
\sum_{i=1}^{M-1}\utility(i)(1-p)^{i-1} &<& P_{3G} - B. \label{cond5_w3g}
\end{eqnarray}
\item $\alpha=2$,
\begin{equation}
U(1)-U(2) \ge P_{3G} -B. \label{cond6_w3g}
\end{equation}
\end{itemize}

\end{proposition}

\begin{proof}
The proof is similar in spirit to the one of Proposition~\ref{propoalin}.
\end{proof}

\begin{remark}
If the optimal threshold policy is unique, the conditions  in Proposition~\ref{allinact_w3g} are necessary and sufficient for the optimal policy to satisfy ~\eqref{cond1_w3g}-~\eqref{cond6_w3g}.
\end{remark}

\subsection{Optimal threshold properties with 3G}

To derive the optimal threshold strategy properties, we distinguish two scenarios: $P_{3G}\leq \frac{G}{p}+P$ and $P_{3G} > \frac{G}{p}+P$.

\bigskip
\begin{center}
\begin{equation}
\boxed{ \boxed{
\textrm{{scenario 1)  } } P_{3G} \leq \frac{G}{p}+P} }
\end{equation}
\end{center}
\bigskip

In this scenario, 
 \begin{equation}
a(x) = \left\{ \begin{array}{ll}  0, & \textrm{ if }  x<s_{3G}\\
2,&  \textrm{ otherwise} \end{array} \right.
\end{equation}
The transition matrix of the system, $\mathcal{P}_{u}$, is
\begin{equation}
\mathcal{P}_{u}=\left(
\begin{array}{lllllllllll}
0  &  1 &  &        &   & \\
0  &  0 & 1&        &   & \\
0  &  0 & 0& 1      &   & \\
   &    &  & \ddots &   & \\
   1  &    &  &        & 0 & 0 & \\
   &    &  &   & & \ddots &   & \\
1  &    &  &   & &      & 0 & 0 & \\
1  &   &   &     & &   &      &  0 \\
 \end{array}
\right)
\end{equation}

Let $\pi$ be the steady state solution of the system for a  threshold $s_{3G}$, $\pi \mathcal{P}_{u}  =  \pi$. Then,

\begin{eqnarray}
\pi_1 &=& \pi_i=\frac{1}{s_{3G}}, \quad i=1, \ldots, s_{3G} \\
\pi_{i} &=& 0, \quad i=s_{3G}+1, \ldots, M\\
\end{eqnarray}
The expected age, $A(s_{3G})=\sum_{i=1}^{M} \pi_i i$, is
 \begin{equation}
A(s_{3G}) = \left\{ \begin{array}{ll}  {(s_{3G}+1)}/{2}, & \textrm{ if }  1\leq s_{3G} \leq M\\
M, &  \textrm{ if } s_{3G} = M+1\end{array} \right.
\end{equation}
The average reward of this threshold policy is 
\begin{equation}  \boxed{
E[r;s_{3G}]=  \left\{ \begin{array}{ll}  \frac{1}{s_{3G}} \left[ \sum_{i=1}^{s_{3G}} \utility(i)- \cost -p P-(1-p) P_{3G} +B\right], & \textrm{ if }  1\leq s_{3G} \leq M\\
0, &  \textrm{ if } s_{3G} = M+1.\end{array} \right. }
\end{equation}

The optimal threshold value $s_{3G}^{\star}$ is
\begin{equation}
s^{\star}_{3G} = \min\left\{ s \Big| \sum_{i=1}^s d(i) -s \utility(s+1) \geq G+p P + (1-p) P_{3G} \right\}
\end{equation}

\bigskip
\begin{center}
\begin{equation}
\boxed{ \boxed{
\textrm{{scenario 2)  } } P_{3G} > \frac{G}{p}+P} }
\end{equation}
\end{center}
\bigskip

 In this scenario, let $s_W = s$, 
 \begin{equation}
a(x) = \left\{ \begin{array}{ll}  0 & \textrm{ if }  x<s_{W}\\
1 & \textrm{ if }  s_{W}\leq x <s_{3G}\\
  2&  \textrm{ otherwise} \end{array} \right.
\end{equation}
The case $s_{3G} = M+1$ was considered in Appendix \ref{sec:no3g}. Next, we  consider the case
$s_{3G} \leq M$. The transition matrix of the system, $\mathcal{P}_{u}$, is
\begin{equation}
\mathcal{P}_{\pi}=\left(
\begin{array}{lllllllllll}
0  &  1 &  &        &   & \\
0  &  0 & 1&        &   & \\
0  &  0 & 0& 1      &   & \\
   &    &  & \ddots &   & \\
   p  &    &  &        & 0 & 1-p & \\
   &    &  &   & & \ddots &   & \\
p  &    &  &   & &      & 0 & 1-p & \\
1  &   &   &     & &   &      &  0 \\
1  &   &   &     & &   &      &  0 \\
1  &   &   &     & &   &      &  0 \\
 \end{array}
\right)
\end{equation}
The  steady state solution, $\pi$, of the above system is
\begin{eqnarray}
\pi_1 &=& \pi_i, \quad i=1, \ldots, s_W \label{eqH1}\\
\pi_{i} &=& \pi_{i-1} (1-p)=\pi_1 (1-p)^{i-s_W}, \quad i=s_W+1, \ldots, s_{3G}\label{eqH2} \\
\pi_{i} &=& 0,  \quad i=s_{3G}+1, \ldots, M \\
\sum_{i=1}^M \pi_i &=& 1
\end{eqnarray}
To obtain $\pi_1$, we sum \eqref{eqH1} and \eqref{eqH2},
\begin{eqnarray} \label{yieldsp1}
s_W  \pi_1 +  \pi_1 \sum_{i=s_W+1}^{s_{3G}}  (1-p)^{i-s}   &=&1 
\label{exavgre1}
\end{eqnarray}
therefore,
\begin{eqnarray}
\pi_1&=&\frac{1}{s_{W}-1+(1-(1-p)^{s_{3G}-s_{W}+1})p^{-1} }.
\end{eqnarray}
The expected reward is
\begin{equation} \label{eq:ers_w3g} \boxed{
E[r;s_W,s_{3G}]=\pi_1   \left[ \sum_{i=1}^{s_W-1} \utility(i) + \sum_{i=s_W}^{s_{3G}} (1-p)^{i-s_W}\utility(i)
- (G/p+P-B) - (P_{3G}-P-G/p)(1-p)^{s_{3G}-s_{W}+1}\right] } .
\end{equation}
The optimal thresholds $(s_W^{\star},s_{3G}^{\star})$ are 
$(s_W^{\star},s_{3G}^{\star}) = {\arg\max}_{(s_W,s_{3G})} \{E[r;s_W,s_{3G}]\}$. 

\subsection{On the number of optimal thresholds}

If there exists an  optimal policy which consists of switching being inactive and WiFi, the results presented in \textsection\ref{sec:appdiscussion} hold if we restrict to policies for which $a(x) \in \{0, 1\}$ ($x=1, \ldots, M$).  In contrast,  if there is an  optimal policy which consists of adopting 3G at state $s_{3G} < M$, the optimal policy is in general not unique, since changing the actions chosen at states $s > s_{3G}$ does not compromise the optimality of the resulting policy.  A characterization of the number of optimal two-threshold policies in this context is subject of future work.

\newpage

\pagebreak
\clearpage

\section{Learning algorithm} 

\subsection{Proof of Proposition~\ref{proposition:converge}}
In what follows, we assume that the {{{\ensuremath{\mbox{\sc  Publisher Problem}}}}} admits a solution $B^{\star}$, and that $B^{\star} > 0$.

\begin{proof}
We index  the discrete rounds of Algorithm~\ref{alg:algorit} by  $n$ 
$(n=1, 2, 3, \ldots)$,
\begin{eqnarray} \label{differenceeq0}
B_{n+1} &\leftarrow &\min\left(\widehat{B}, \max\left(0, B_n + 
\frac{\alpha}{n} \left( T- Q_n\right)\right) \right)
\end{eqnarray}

Then, we prove a generalized  version of  Proposition~\ref{proposition:converge}.   Namely, we consider the following algorithm, obtained from~\eqref{differenceeq0} after removing the $\max$ and $\min$ operators,
\begin{eqnarray}  \label{differenceeq}
{B}_{n+1} &\leftarrow &  {B}_n + \frac{\alpha}{n}  (T-Q_n)
\end{eqnarray}
Next, we show that~\eqref{differenceeq} converges  to the optimal solution of the {{{\ensuremath{\mbox{\sc  Publisher Problem}}}}} with probability~1. \eat{  We  also show that the  basin of attraction of~\eqref{differenceeq} is unique.}  If $\widehat{B}$ is sufficiently large,  the convergence of~\eqref{differenceeq} to the optimal solution of the {{{\ensuremath{\mbox{\sc  Publisher Problem}}}}} implies the convergence of~\eqref{differenceeq0} to such a solution as well (see~\cite[Lemma 3.3.8 and \textsection 5.4]{borkar} for details).

\subsubsection{Martingale definition} Let $M_{n+1}= E[Q_n]-Q_n$ and $g(B)=T-E[Q_n(B)]$.  Then
\begin{eqnarray} \label{eq:aporiginal1}
B_{n+1} &\leftarrow & B_n + 
\frac{\alpha}{n}  \left(g(B_n)+M_{n+1}\right)
\end{eqnarray}
The history of values $\{B_m\}$ 
and $\{M_m\}$, $m=1, \ldots, n$, yields a
$\sigma$-field ${\cal F}_n$, $  {\cal F}_n =\sigma(B_m,M_m, m\leq n)$. 

Note that $\{B_m\}$, $m=2, \ldots, n$ are fully characterized by $B_1$ 
and  $\{M_m\}$, $m=1, \ldots, n$. Therefore $  {\cal F}_n =\sigma(B_1,M_m, m\leq n)$.

Moreover, $\{M_n\}$, $n=1,2,\ldots$ is a sequence of zero mean random 
variables satisfying
$$
E\left[M_{n+1}| {\cal F}_n\right]=0\; a.s., n\geq 0.
$$
As a consequence, $\{M_n\}$ is a martingale difference sequence  with 
respect to
the increasing $\sigma$-fields ${\cal F}_n$, $n=1,2,\ldots$.

Recall that $N$ is the number of mobiles in the network. Note that
\begin{equation}
Q_n \leq N
\end{equation}
Therefore, $|M_{n+1}|=|E[Q_n]-Q_n| \leq N$, thus $\{M_n\}$ is square-integrable, i.e.,
\begin{equation} \label{squareint}
E[|M_{n+1}|^2|{\cal{F}}_n] \textrm{ is finite}
\end{equation}

\bigskip
\subsubsection{Differential inclusion definition}
Since the function $g(B)=T-E[Q_n(B)]$ is discontinuous, we use 
differential inclusions, an extension of the ordinary differential 
equations, to approximate~\eqref{differenceeq}.  The solution concept 
adopted for the differential inclusions is the  \emph{Filippov 
solution}~\cite{filippov2}.  \eat{fillipov,stability}.

Let $s_0$ be the number of discontinuity points of $g$ in the interval 
$(0,\widehat{B}]$.
Let $B^{(i)}$ be the value of $B$ at the  $i^{th}$ discontinuity point 
of $g(B)$, $i=0, 1, 2, \ldots, s_0$.
Then, $B^{(0)} =0< B^{(1)} < \ldots < B^{(s_0)}=\widehat{B}$ (see 
Figure~\ref{fig:convergence}).

Next, we approximate~\eqref{differenceeq}  by the following differential 
inclusion
\begin{equation}
\dot B(t)\in F(B(t))
\label{DI}
\end{equation}
where $F$ is a set-valued map
$$F: [0,\widehat{B}]\to \mbox{subset of } \mathbb{R}$$
$F$ is defined as follows
\begin{equation}
F(B)=\left\{
\begin{array}{ll}
\{g(B)\}, & \mbox{ if } B\not= B^{(i)}, \qquad i=1, \ldots, s_0 \\
\;[g(B^{(i-1)}), g(B^{(i)})], &
\mbox{ if } B=B^{(i)}, \qquad i=1, \ldots, s_0
\end{array}
\right.
\end{equation}

Since the function $F$ is semi-continuous and, for all $B(t)$, $F(B(t))$ 
is a compact and a convex set,  there exists a solution to the 
differential inclusion \eqref{DI}.  In addition, it follows from Theorem 
2 in \cite{filippov2} \eat{(see also~\cite{filippov}) } that such a solution is 
unique.

Let $B^{\star}$ be the following solution to the {{\ensuremath{\mbox{\sc  Publisher Problem}}}},  
\begin{equation}
B^{\star}=\max\Big\{B^{(i)} \Big| g\left(B^{(i)}\right) \ge  0\Big\}
\label{sol}
\end{equation}
Note that $0\in [g(B^{(i)})=g(B^{\star}), g(B^{(i+1)})]$ and let 
\begin{equation}
\widetilde{T} = \max\Big\{ t \le T \Big| t - \frac{N}{s(B^{\star})+(1-p)/p} = 0 \Big\}
\end{equation}
The \emph{basin of attraction} of~\eqref{DI} is  (see Figure~\ref{fig:convergence}),
\begin{eqnarray*}
\left\{
\begin{array}{ll}
[B^{(i)}=B^{\star}, B^{(i+1)}], &  \textrm{if  } T = \widetilde{T} \\
B^{(i+1)}, &  \textrm{if  }  T \neq \widetilde{T}, \textrm{ where } B^{(i)}=B^{\star}
\end{array}
\right.
\end{eqnarray*}

\bigskip
  We establish the global stability of the basin of attraction 
of~\eqref{DI}  in~\textsection\ref{learning:sub2}, and then 
in~\textsection\ref{learning:sub3}  we use~\cite[Theorem 2, 
\textsection{5}]{borkar} to show that the algorithm \eqref{differenceeq} 
converges to the basin of attraction of~\eqref{DI}.

\subsubsection{Global stability of basin of attraction} 
\label{learning:sub2}
In order to establish the global stability of the basin of attraction of~\eqref{DI}, we use the following Lyapunov function
\begin{equation}
V({B})=({B}-B^{*})^2
\end{equation}

We begin by stating two definitions which generalize the notions of 
gradients and derivatives to encompass differential 
inclusions~\cite{shevitz}.

\begin{definition}
The Clarke generalized gradient of $V$ in $B$ is defined as
\begin{eqnarray}
\partial V(B)&=& co\{\lim_{l\to \infty} \Delta V(B_l) : (B_l)\to (B)\}
\end{eqnarray}
where $co\{A\}$ is the smallest convex and bounded set containing $A$.
\end{definition}

\begin{definition}
The set-valued derivative of $V$ is  defined as
\begin{eqnarray*}
\dot{V}({B})&=&\{a\in \mathbb{R}\;: \exists v \in F({B}) \mbox{ so that 
} pv=a, \;\; \forall p\in \partial V({B})\}
\end{eqnarray*}
\end{definition}

It follows from the two definitions above that
\begin{eqnarray}
\partial V(B)&=&\{ 2({B}-B^{*})\}
\end{eqnarray}
and
\begin{eqnarray*}
\dot{ V}({B})
&=& \left\{
\begin{array}{ll}
\{2({B}-B^{\star}) g(B)\},&\mbox{ if } B\not=B^{(j)}, \qquad j=1, \ldots, s_0\\
\;[2({B}-B^{\star}) g(B^{(j-1)}), 2({B}-B^{\star}) g(B^{(j)})], & \mbox{ if } B=B^{(j)}
\end{array}
\right.
\end{eqnarray*}

We divide the analysis of the stability of the basin of attraction of~\eqref{DI} into two cases, varying according to whether $T \neq \widetilde{T}$ (see Figure~\ref{fig:convergence}(b)) or $T = \widetilde{T}$ (see Figure~\ref{fig:convergence}(c))
\begin{itemize}
\item If $T \neq \widetilde T$ (see Figure~\ref{fig:convergence}(b)), $\dot{ 
V}({B})<0$ for $B\not=B^{\star}$, 
\begin{itemize}
\item \underline{$B \neq B^{(j)}$}:  If $B > B^{\star}$, then $g(B) = 
T-E[Q_n(B)] < 0$ since $E[Q_n]$ is an increasing function of $B$ (see 
Proposition~\ref{increasing}). If $B < B^{\star}$ the result follows 
similarly;
\item \underline{$B = B^{(j)}$ and $B\neq B^{\star}$}:  If $B > 
B^{\star}$, then $g(B^{(j-1)}) = T-E[Q_n(B^{(j-1)})] < 0$ since $E[Q_n]$ 
is an increasing function of $B$ (see Proposition~\ref{increasing}) and 
$B^{(j-1)} \leq B^{(j)}$. If $B < B^{\star}$ the result follows similarly.
\end{itemize}
\item Using similar arguments, we can establish the stability of the basin of attraction when $T = \widetilde T$ (see 
Figure~\ref{fig:convergence}(c)).  Let  $\underline{B^{\star}}=B^{\star}=B^{(i)}$ and  $\overline{B^{\star}}=B^{(i+1)}$. \eat{ Note that $g(\underline{B^{\star}})=0$.   Hence,} Then,   $\dot{ V}({B})<0$ 
for $B\not \in [g(\underline{B^{\star}}), g(\overline{B^{\star}})]$.
\end{itemize}
\medskip

Therefore, the global stability of the basin of attraction of~\eqref{DI} 
follows from the Lyapunov theorem (see ~\cite{borkar}). \eat{ and khalil ~\cite{vidya}).}

\bigskip
\subsubsection{Algorithm convergence to differential inclusion}  Next, our goal is to show that 
~\eqref{differenceeq} is well approximated by ~\eqref{DI}. To this goal, we define the linear interpolation of $B_n$ and establish the finiteness of $\sup_n |B_n|$.

\emph{Interpolation of $\{B_n\}$: } 
\label{learning:sub3}
   Let  $\bar B(t)$ be the linear interpolation of $B_n$, $n=1, 2, 3, \ldots$
   \begin{equation}
   \bar B(t)=  \label{eqlimitpoint}
   \left\{
\begin{array}{ll}
B_n, &  t=\sum_{i=1}^n  {\alpha}/{i} \\
B_n + \displaystyle \frac{t-\sum_{i=1}^n  {\alpha}/{i} 
}{\sum_{i=1}^{n+1}  {\alpha}/{i}- \sum_{i=1}^{n}  {\alpha}/{i}}  
(B_{n+1} - B_n), & t \in (\sum_{i=1}^n  {\alpha}/{i},\sum_{i=1}^{n+1}  
{\alpha}/{i})
\end{array}
\right.
\end{equation}

$\bar{B}(t)$ is a piecewise linear and continuous
function.

\emph{Finiteness of $\sup_n|B_n|$: } Next, we show that  $\sup_{n} |B_n|<\infty$.   For $B\in[0,\infty)$, $F(B)$ is  a compact and convex set and 
\begin{equation}
\sup_{y\in F(B)} |y|< g(B^{(0)})<  g(B^{(0)})(1+ B)
\label{bound}
\end{equation}
Let $h_c(B)$ be defined as in \cite[\textsection 3.2-(A5)]{borkar}, $h_c(B) = g(cB)/c =   ({T -  N/(s(cB)+(1-p)/p)   })/{c}$. Then, for   $B\in[0,\infty)$,
\begin{equation}
h_{\infty}(B) = 0 \label{hinf}
\end{equation}
Therefore, similar arguments as those in \cite[\textsection 3.2]{borkar}  applied to \eqref{squareint}, \eqref{bound}  and \eqref{hinf}    yield
\begin{equation}
\sup_{n} |B_n|<\infty
\label{stab}
\end{equation}

\emph{Concluding the proof: } 
Given that \eqref{stab} (which corresponds to~\cite[(5.2.2)]{borkar}) holds, Theorem 2  in 
\cite[\textsection 5]{borkar} establishes that  almost surely every
limit point of ~\eqref{eqlimitpoint}
satisfies \eqref{DI} (see~\cite{borkar} for details). As a consequence,  
since the sequence $\{B_n\}_{n=1,2,\ldots}$ generated 
by~\eqref{differenceeq} is contained in~\eqref{eqlimitpoint}, it 
converges almost surely to ~\eqref{DI} (see Corollary 4 of 
\cite[\textsection 5]{borkar} for details).
\begin{figure}
\center
\includegraphics[scale=0.55]{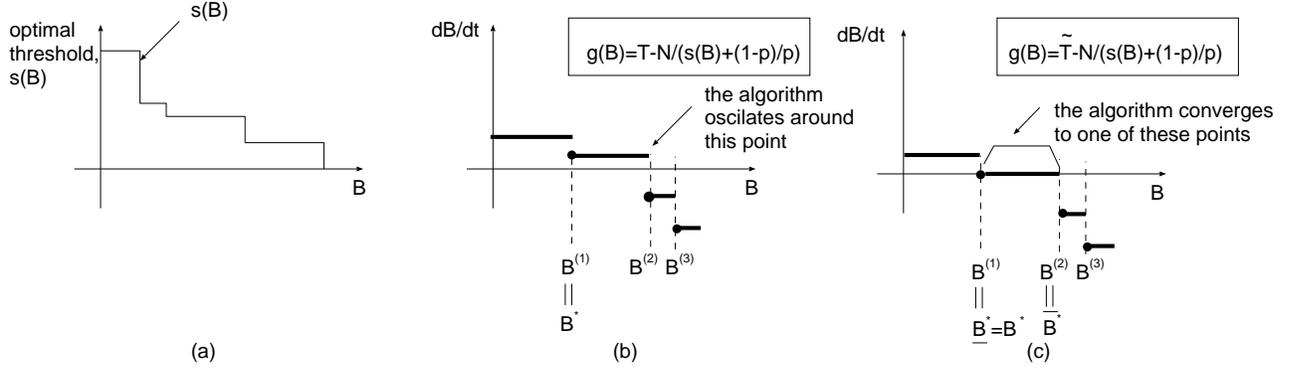}
\caption{Convergence of learning algorithm. $(a)$ the threshold policy 
adopted by users is piecewise constant; $(b)$ $g(B)$ when $T \neq 
\widetilde{T}$; $(c)$ $g(B)$ when $T = \widetilde{T}$} 
\label{fig:convergence}
\end{figure}
\end{proof}

\subsection{Binary search and stochastic learning}

In this paper we proposed the use of binary search to obtain the optimal bonus level when the publishers have complete information~(see Proposition~\ref{optimalpubprob}) and stochastic learning in face of incomplete information~(see Proposition~\ref{proposition:converge}).  A variation of binary search, as proposed in~\cite{laurent}, can also be used in face of incomplete information.  However, we do not have a proof of the convergence of the modified  binary search algorithm in our setting.  Since the binary search algorithm may have faster convergence to the optimal solution, it can be used to obtain an initial value for the stochastic learning, whose convergence was established in~Proposition~\ref{proposition:converge}.  A comparison of pros and cons of the modified binary search and stochastic learning is  provided in~\cite{laurent}.

\subsection{Simulation }

\label{app:simulation}

In order to 1) study the behavior of our algorithm when the population size varies and to 2) investigate how the convergence speed may be affected by the correlation among users, we conducted trace-driven simulations, whose results  are   shown 
in~\url{http://www-net.cs.umass.edu/~sadoc/agecontrol/bus-ap-ct20/index.html}.

\paragraph{Simulation setup}

We set the parameters of our learning algorithm as follows: $M$=30, $G$=0.4, $T$=11, $P$=100 and  $\tau=10$ time slots. The number of users is initially 105  and drops to 90 at round 100. Finally, $\alpha$=10 for trace driven simulations, and $\alpha=20$ under the uniformity and independence assumption.

In our trace drive simulations in this section, we  consider half of the population in one bus and the other half in another.  Recall that for each bus shift we generate a  string of zeros and ones corresponding to slots with and without a useful contact opportunity, respectively.  In order to distinguish the users in two buses, while simulating the opportunities observed by each user we assume that half of them  take their first observations in the beginning of the string of zeros and ones while the other half take their first observations in the middle of the trace.  After taking their first observations, the subsequent ones are drawn in sequence out of the string of zeros and ones (our simulator is available at~\url{http://www-net.cs.umass.edu/~sadoc/agecontrol/learning.tgz}).

Despite the correlations among users,  for the 88 bus shifts analyzed our trace driven  results did not significantly deviate from the ones  obtained with  uncorrelated users.   In \url{http://www-net.cs.umass.edu/~sadoc/agecontrol/bus-ap-ct20/index.html} we show for each bus shift the results obtained with trace driven simulations as well as under the uniformity and independence assumption (the contact opportunity being obtained from traces).   In what follows, we illustrate our results with two examples.

\paragraph{Simulation Results}

Figures~\ref{5figsA} and ~\ref{5figsB} illustrate our simulation results for two bus shifts.  Figures~\ref{5figsA}(a) and ~\ref{5figsB}(a)  show the contact opportunities.  Note that while in Figure~\ref{5figsA}(a)  the opportunities are roughly uniform across different time slots, in Figure~\ref{5figsB}(a) there is a high concentration of opportunities close to Haigis Mall, but not many in-between arrivals at Haigis Mall.  This, in turn, impacts the results obtained with the learning algorithms, as shown in Figures~\ref{5figsA}(b)-(c) and ~\ref{5figsB}(b)-(c).  The dotted lines represent the optimal range of bonuses.  While in Figure~\ref{5figsA}(b) the bonus computed using the learning algorithm converges to optimal values, in Figure~\ref{5figsB}(b) the difference between the optimal bonus and the one obtained with the learning algorithm was around 20.  Despite such difference, though, note that the number of transmissions experienced by the service providers oscillated around its target, eleven, in both cases (see Figures~\ref{5figsA}(c) and Figures~\ref{5figsB}(c)).  

In Figure~\ref{5figsA}(c) the number of transmissions remains stable during certain intervals of time (for instance, in the interval [20,30]).  This is partly due to the synchronization of the users, that are assumed to begin with the same initial state and to be experience correlated contact opportunities.       In contrast, if we consider users that experience contact opportunities uniformly at random, such synchronizations do not occur and the bonus computed using the learning algorithm converges to optimal values in both Figure~\ref{5figsA} and~\ref{5figsB} (see Figures~\ref{5figsA}(d)-(e) and Figures~\ref{5figsB}(d)-(e)). 

Further insight into the synchronization among users in Figures~\ref{5figsA}(c) is obtained from Figure~\ref{samplepath}.   Each curve in Figure~\ref{samplepath} shows the age of the users in each of the buses. Since all the users are assumed to begin with age 1 at slot 1, the age of all users in a bus remains equal throughout our simulations.  Figure~\ref{samplepath} shows that from slots 300 to 400 (which correspond to rounds 30 to 40), in each round each user issues one update.  Since there are 105 users, the average number of updates per slot is 10.5.  Note that such synchronization does not occur between slots 400 and 450, when the number of updates issued by a user in a given slot varies between 1 and 2.  The synchronization between users is an artifact due to the assumption that they have the same initial state at slot 1, and does not occur if users have their initial states sampled uniformly at random.

\begin{figure*}
\center
\includegraphics[scale=0.4]{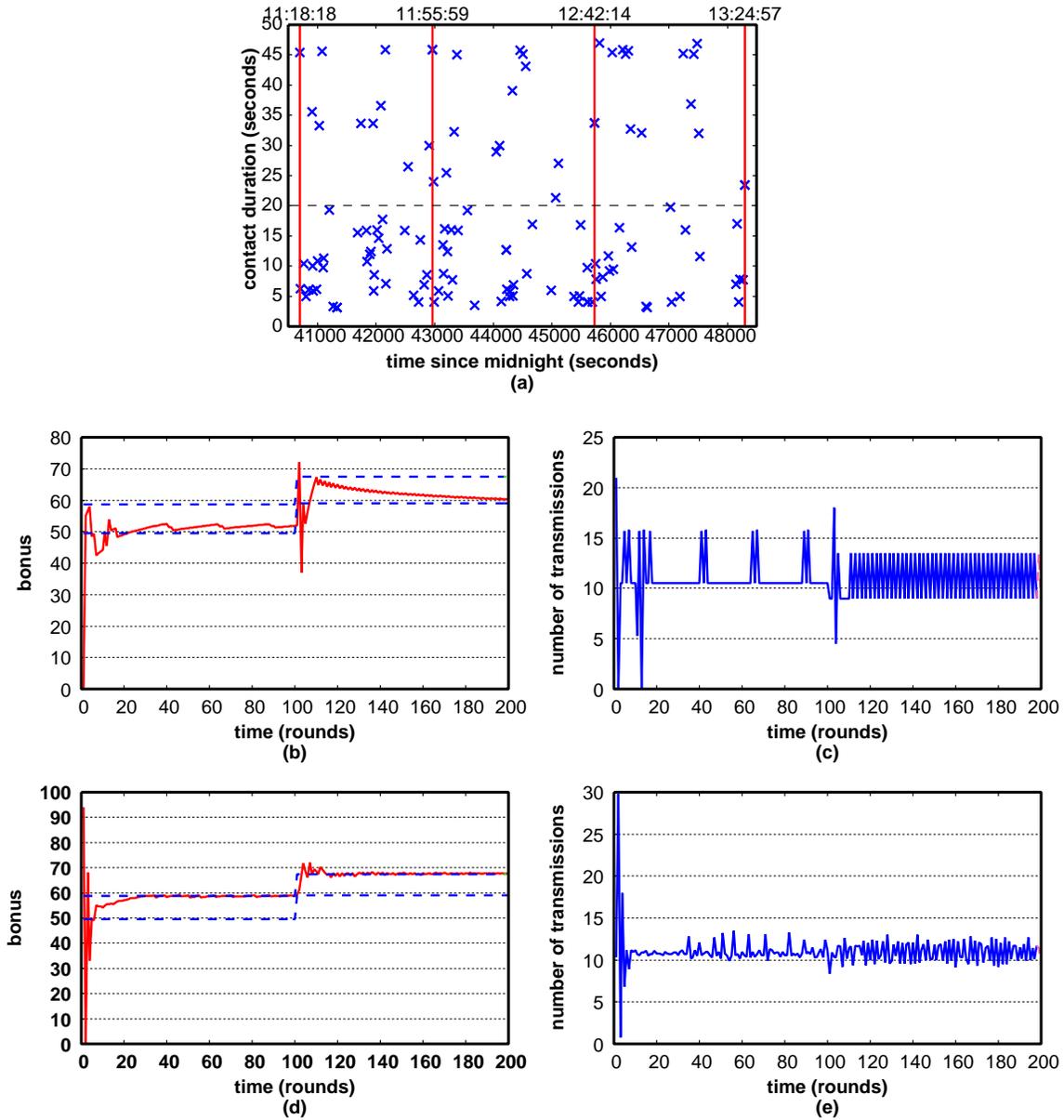}
\caption{(a) Contact opportunities; (b) and (c) show the bonus  and number of transmissions obtained from trace driven simulations; (d) and (e) show the bonus  and number of transmissions obtained from simulations under uniformity and independence assumptions. }
\label{5figsA}
\end{figure*}

\begin{figure}
\center
\includegraphics[scale=0.8]{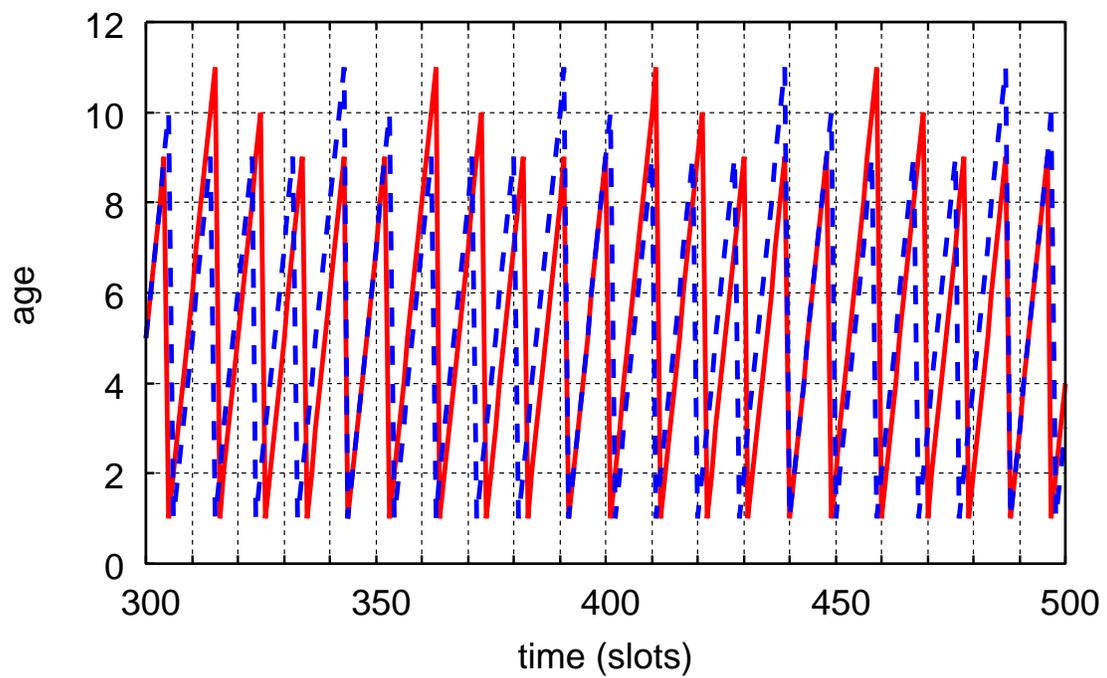}
\caption{Sample path of age as a function of slot (trace driven simulation) See also Figure~\ref{5figsA}(c).}
\label{samplepath}
\end{figure}

\begin{figure*}
\center
\includegraphics[scale=0.4]{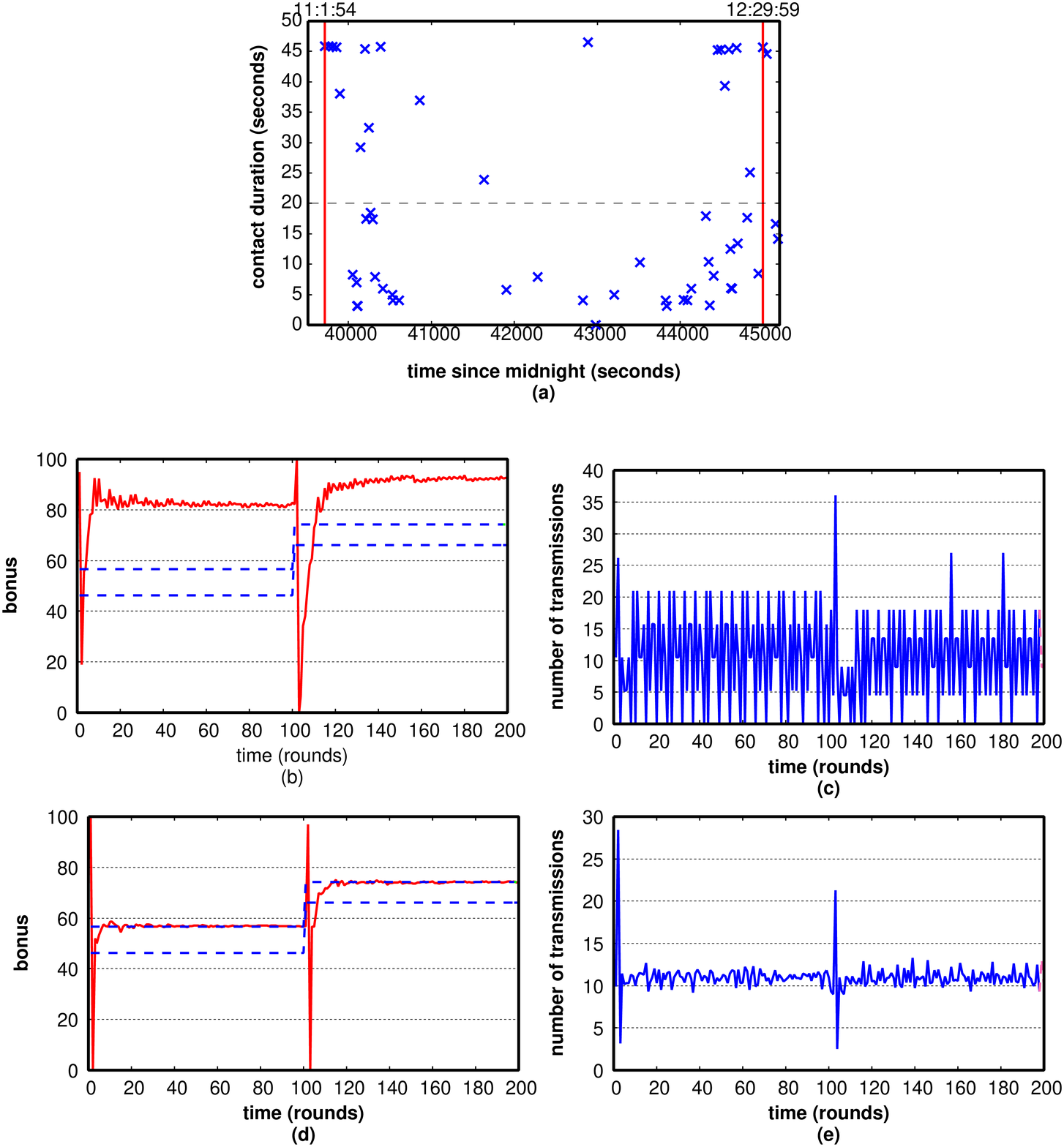}
\caption{(a) Contact opportunities; (b) and (c) show the bonus  and number of transmissions obtained from trace driven simulations; (d) and (e) show the bonus  and number of transmissions obtained from simulations under uniformity and independence assumptions. }
\label{5figsB}
\end{figure*}

\begin{figure*}
\includegraphics[scale=0.45]{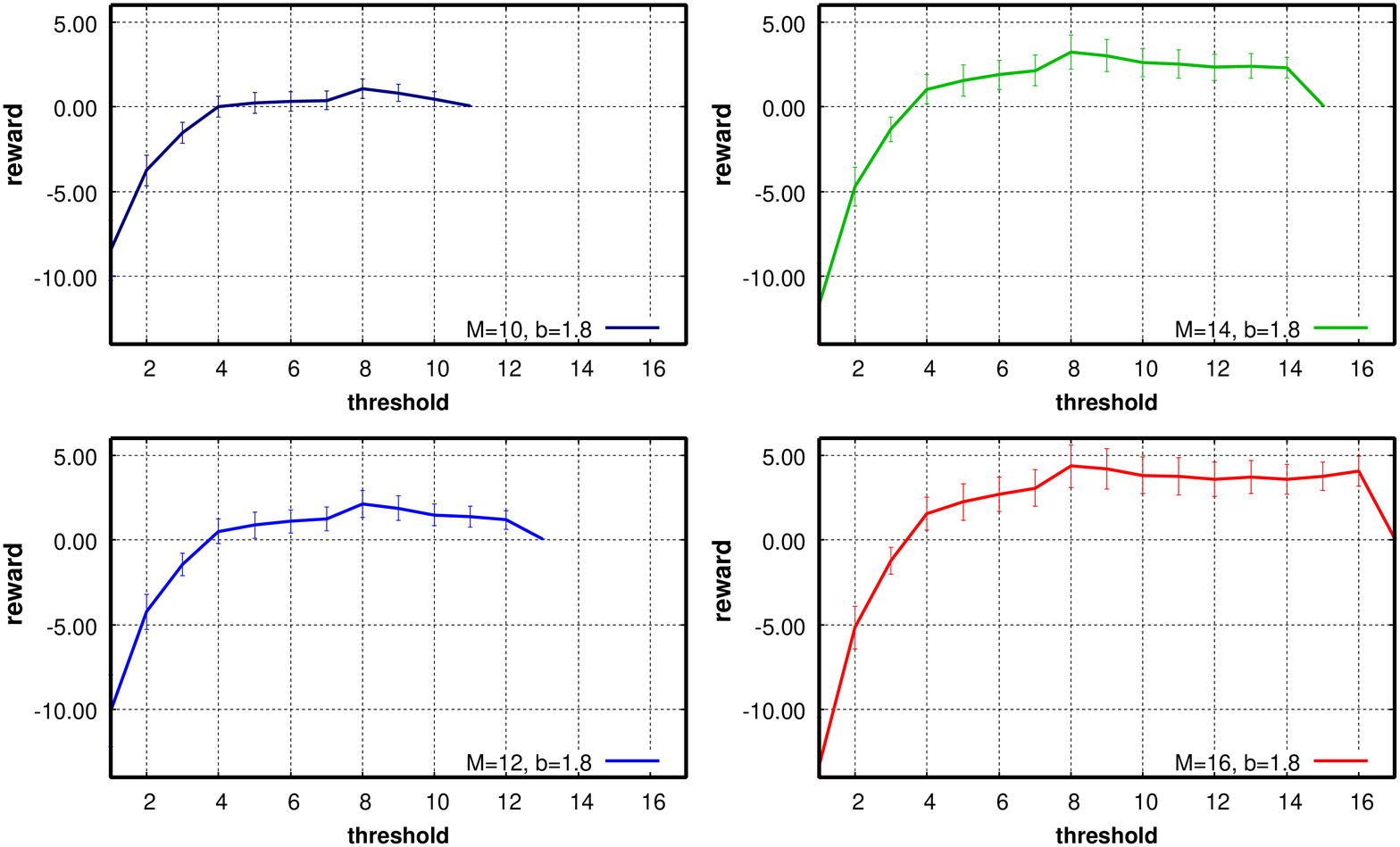}
\caption{Trace-driven reward as a function of the threshold, where threshold is assumed the same at  all bus shifts (Figure~\ref{utility1a}(c) with 95\% confidence intervals)}
\label{4figs95int}
\end{figure*}

\end{appendix}

\end{document}